\documentclass[preprint,authoryear,11pt]{elsarticle}
 \usepackage{amsmath,amssymb,amsthm,graphicx,caption}
 \usepackage{amsfonts}
 \usepackage{mathrsfs}
 \usepackage{verbatim}
 \usepackage[usenames, dvipsnames]{color}
 \usepackage[capposition=top]{floatrow}
 \usepackage[section]{placeins}
 \usepackage[utf8]{inputenc}
 \usepackage{rotating}
 \usepackage[colorlinks=true,linkcolor=blue]{hyperref}%
 \usepackage{lscape}
 \usepackage{booktabs}
 \usepackage{bm}
 \usepackage{ulem}	
 \usepackage{tcolorbox}
 \usepackage{mathtools}
 \usepackage{autobreak}
  
  \newcommand{\Lower}[1]{\smash{\lower 1.5ex \hbox{#1}}}

 \numberwithin{equation}{section}

  \def\vep{\varepsilon}
  
  \def\R{\mathbb{R}}
  \def\bD{\bm{D}}
  \def\bX{\bm{X}}
  \def\bY{\bm{Y}}
  \def\bZ{\bm{Z}}
  \def\bv{\bm{v}}
  \def\bs{\bm{s}}
  
  \def\bbt{\bm{\beta}}
  \def\balpha{\bm{\alpha}}
  \def\bphi{\bm{\phi}}

  \newenvironment{eqarray}{\arraycolsep 0.14em\begin{eqnarray}}{\end{eqnarray}}
  \newenvironment{eqarray*}{\arraycolsep 0.14em\begin{eqnarray*}}{\end{eqnarray*}}
  \newtheorem{thm}{Theorem}
  \numberwithin{thm}{section}
  \newtheorem{lem}[thm]{Lemma}

  \numberwithin{equation}{section}

 \journal{Submitted to TBA}

\begin{document}
\allowdisplaybreaks
\begin{frontmatter}

\title{Extrapolation Estimation for Nonparametric Regression\\ with Measurement Error}
\author{Weixing Song, Kanwal Ayub}
\address{Department of Statistics, Kansas State University, Manhattan, KS 66506}
\author{Jianhong Shi\corref{cor}}
\address{School of Mathematics and Computer Sciences, Shanxi Normal University, Linfen, China 041000}
\ead{weixing@ksu.edu}
\cortext[cor]{Corresponding author}

\begin{abstract}
  For the nonparametric regression models with covariates contaminated with normal measurement errors, this paper proposes an extrapolation algorithm to estimate the nonparametric regression functions. By applying the conditional expectation directly to the kernel-weighted least squares of the deviations between the local linear approximation and the observed responses, the proposed algorithm successfully bypasses the simulation step needed in the classical simulation extrapolation method, thus significantly reducing the computational time. It is noted that the proposed method also provides an exact form of the extrapolation function, but the extrapolation estimate generally cannot be obtained by simply setting the extrapolation variable to negative one in the fitted extrapolation function if the bandwidth is less than the standard deviation of the measurement error. Large sample properties of the proposed estimation procedure are discussed, as well as simulation studies and a real data example being conducted to illustrate its applications.
\end{abstract}

\end{frontmatter}



 \section{Introduction}\label{sec1}

  Due to its conceptual simplicity and the capability to harness the modern computational power, the simulation extrapolation estimation (SIMEX) procedure has been attracting significant attention from practical data analysts as well as theoretical researchers. The simplicity of the SIMEX lies in the fact that it allows us to directly use any standard estimates based on the known data as the building block, and its simulation nature makes the estimation process computer-dependent only. To be specific, suppose we want to estimate a parameter $\theta$, possibly multidimensional, in a statistical population $X$ of dimension $p$, where $p\geq 1$. In certain situations where we cannot collect observations directly from $X$, what we observe is a surrogate value $Z$ of $X$. In measurement error literature, a classical assumption on the relationship between $X$ and $Z$ is $Z=X+U$, where $U$ is called the measurement error, which is often assumed to be independent of $X$, and has a normal distribution with mean $0$ and known covariance matrix $\Sigma_u$. If there is an estimator $T(\bX)$ of $\theta$ when a sample $\bX=\{X_1,\ldots, X_n\}$ of size $n$ from $X$ is available, then when only $Z$ can be observed, the classical SIMEX procedure estimates $\theta$, using sample $\bZ=\{Z_1,\ldots, Z_n\}$ from $Z$, by going through the following three steps. First, we generate $n$ i.i.d. random vectors $V_i$'s from $N(0,\Sigma_u)$, select a nonnegative number $\lambda$, calculate $\tilde Z_i(\lambda)=Z_i+\sqrt{\lambda}V_i$ for $i=1,2,\ldots, n$, and compute $T(\tilde\bZ(\lambda))$ based on $\tilde{\bZ}(\lambda)=\{\tilde Z_1(\lambda),\ldots, \tilde Z_n(\lambda)\}$. Second, we calculate the conditional expectation of $T(\tilde\bZ(\lambda))$ given $\bZ$. If the conditional expectation has a closed form, then it will be the estimate of $\theta$, otherwise, we repeat the previous step $B$ times to obtain $B$ values of $T_b(\tilde\bZ(\lambda))$, $b=1,2,\ldots, B$, and the average $\bar T(\lambda)$ of these $B$ values of $T_b(\tilde\bZ(\lambda))$'s is computed. Finally, we repeat the first step and second step for a sequence of nonnegative $\lambda$ values, for example, $0=\lambda_1<\ldots<\lambda_K$ for some $K$. We denote these $K$ averages as $\bar T(\lambda_1),\ldots, \bar T(\lambda_K)$. To conclude, the trend of $\bar T(\lambda)$ with respect to $\lambda$ will be formulated as a function of $\lambda$, and the extrapolated value of this function at $\lambda=-1$ is the desired SIMEX estimate of $\theta$. In real applications, $K$ is suggested to be less than $20$ and these $K$ $\lambda$-values are chosen equally spaced from $[0,2]$. The early development of the classical SIMEX estimation procedure can be found in \cite{cook1994}, \cite{stefan1995} and \cite{carroll1996}, with extensive applications in \cite{mallick2002} for cox regression, \cite{sevil2019} for Log-logistic accelerated failure time models, \cite{gould1999} for the catch-effort analysis, \cite{hwang2003}, \cite{stoklosa2016} for the capture-recapture models, \cite{lin1999} for the analysis of the Framingham heart disease data using the logistic regression, \cite{hardin2003} for generalized linear models, and \cite{ponzi2019} for some applications in ecology and evolution, to name a few.

  However, the discussion of the classical SIMEX estimation procedure in the nonparametric setup seems scant in the literature. \cite{stef1996} applied the simulation extrapolation procedure to estimate the cumulative distribution function of a finite population based on the Horvitz-Thompson estimator. Since the conditional expectation of the Horvitz-Thompson estimator with the true variable replaced by the pseudo-data given the observable surrogates has an explicit form, the simulation step can be bypassed. Also, the quadratic function of $\lambda$ is shown to be a reasonable extrapolation function. \cite{car1999} extended the classical SIMEX procedure to the nonparametric regression setup and it was implemented with the local linear estimator. In \cite{car1999}'s work, the three steps in the classical SIMEX procedure are strictly followed. To estimate the unknown variance function in a general one-way analysis of variance model, \cite{carr2008} proposed a permutation SIMEX estimation procedure to completely remove the bias after extrapolation.
  \cite{wang2010} generalized \cite{stef1996}'s method to estimate the smooth distribution function in the presence of heteroscedastic normal measurement errors. Aiming at improving the SIMEX local linear estimator in \cite{car1999}, \cite{stauden2004} introduced a new local polynomial estimator with the SIMEX algorithm. The improvement over the existing estimation procedure is made possible by using a bandwidth selection procedure. Again, \cite{stauden2004}'s method still strictly followed the three-steps in the classical SIMEX.

  Compared to various applications in both the parametric and nonparametric statistical models, the SIMEX procedure developed in \cite{stef1996} and \cite{wang2010} successfully dodged the simulation step, which is the most time-consuming part in the classical SIMEX algorithm. The very reason why their methods work is that the averaged naive estimator from the pseudo-data, conditioning on the observed surrogates, has an explicit limit ready for extrapolation, as the number of pseudo-data sets tends to infinity. Clearly, the strategy used in both references cannot be directly extended to other scenarios where such limits do not have user-friendly forms. In this paper, we will propose a new method, which in spirit is a variant of the classical SIMEX procedure, for estimating the nonparametric regression. The new method can also successfully circumvent the simulation step, and the applicable extrapolation functions can also be found, although still being approximated, based on the true but not usable extrapolation functions derived from the theory.

  \section{Motivating Examples}\label{sec2}

  In this section, we shall discuss two motivating examples which inspired our interest in searching for a more efficient bias reduction estimation procedure in the nonparametric setup. Our ambition is to keep the attractive feature of the extrapolation component in the classical SIMEX algorithm, while at the same time, significantly reducing the computational burden.

  \subsection{Simple linear regression model}

  Let $Y$ and $X$ be two univariate random variables, which obey a simple linear relationship $E(Y|X)=\alpha+\beta X$. Suppose we cannot observe $X$ but we have data on $Z=X+U$ and $U\sim N(0,\sigma_u^2)$ with $\sigma_u^2$ being known. As discussed in \cite{car1999}, for any fixed $\lambda>0$, after repeatedly adding the extra measurement errors, and computing the ordinary least squares slope, the averaged estimator consistently estimates $g(\lambda)=\beta\sigma_X^2/(\sigma_X^2+(1+\lambda)\sigma_u^2)$. Obviously, extrapolating $\lambda$ to $-1$, we have $g(-1)=\beta$. This clearly shows that SIMEX works very well for linear regression model. In fact, in the seminal paper \cite{cook1994}, the SIMEX estimators of $\alpha$ and $\beta$ can be derived without the simulation step. However, the derivation relies on a notion of NON-IID pseudo-errors. More details about the NON-IID pseudo-errors can be found in \cite{cook1994} and Section 5.3.4.1 in \cite{carr2006}. Here we would like to point out that the SIMEX estimators of $\alpha$, $\beta$ can be obtained without using the NON-IID pseudo-errors.

  Recall that the least squares (LS) estimator of $\alpha$ and $\beta$ can be obtained by minimizing the LS criterion
  $\sum_{i=1}^n(Y_i-\alpha-\beta X_i)^2$. Since $X_i$ are not available, following the SIMEX idea, we generate the pseudo-data $Z_i(\lambda)=Z_i+\sqrt{\lambda}V_i$, $i=1,2,\ldots,n$. However, instead of following the classical SIMEX road map to minimize the LS target function $\sum_{i=1}^n(Y_i-\alpha-\beta Z_i(\lambda))^2$, we minimize the conditional expectation $E\left[\sum_{i=1}^n(Y_i-\alpha-\beta Z_i(\lambda))^2\Big|\bD\right]$,
  where $\bD=(\bY,\bZ)$,  $\bY=(Y_1,\ldots, Y_n)$ and $\bZ=(Z_1,\ldots, Z_n)$. Since $V_i$'s are i.i.d. from $N(0,\sigma_u^2)$ and independent of other random variables in the model, so this conditional expectation equals
  $\sum_{i=1}^n(Y_i-\alpha-\beta^T Z_i)^2+n\lambda\bbt^T\Sigma_{U}\bbt$.
  The minimizer of the above expression is simply $
  \hat\beta(\lambda)=(S_{ZZ}+\lambda\Sigma_{U})^{-1}S_{YZ}$ and $\hat\alpha(\lambda)=\bar Y-\hat\beta^T(\lambda)\bar X$
  and by choosing $\lambda=-1$, we immediately have the commonly used bias-corrected estimators or the SIMEX estimators derived using NON-IID pseudo-errors. Note that here not only do we not need the simulation step, but also the extrapolation step is unnecessary.

  \subsection{Kernel density estimation}

   Suppose we want to estimate the density function $f_x(x)$ of $X$ in the measurement error model $Z=X+U$. When observations can be made directly on $X$, the kernel density estimation procedure is often called on for this purpose. Starting with the classical kernel estimator, \cite{wang2009} followed the classical SIMEX algorithm, constructed an average of the kernel estimator $\hat f_{B,n}(x)=B^{-1}\sum_{b=1}^B[n^{-1}\sum_{i=1}^nK_h(x-Z_i-\sqrt{\lambda}V_{i,b})]$ with $B$ pseudo-data sets $\{Z_i+\sqrt{\lambda}V_{i,b}\}_{i=1}^n$, $b=1,2,\ldots,B$, where $K_h(\cdot)=h^{-1}K(\cdot/h)$. By the law of large numbers, $\hat f_{B,n}(x)\to n^{-1}\sum_{i=1}^n\int K_h(x-Z_i-\sqrt{\lambda}\sigma_uu)\phi(u)du=\tilde f_n(x)$ in probability. After some algebra, \cite{wang2009} proposed to estimate $f_x(x)$ using $\hat f_n(x)=n^{-1}\sum_{i=1}^n (\sqrt{\lambda}\sigma_u)^{-1}\phi((x-Z_i)/\sqrt{\lambda}\sigma_u)$ which approximates the limit $\tilde f_n(x)$ for sufficiently large $n$. In fact, before initiating the simulation step, \cite{cook1994} suggested one should try to calculate the conditional expectation $E[f_{B,n}(x)|\bZ]$ first. If this conditional expectation has a tractable form, then it will be chosen as the SIMEX estimator. Clearly, the conditional expectation is simply $\tilde f_n(x)$. It is interesting to note that if we deliberately choose the kernel function $K$ to be standard norm, we can show that
   $\tilde f_n(x)=(n\sqrt{\lambda\sigma_u^2+h^2})^{-1}\sum_{i=1}^n \phi((x-Z_i)/\sqrt{\lambda\sigma_u^2+h^2})$ which can also be directly used for extrapolation. Because there is no approximation done here, $\tilde f_n(x)$ should potentially perform better than the estimator $\hat f_n(x)$ as proposed in \cite{wang2009}.

   It is easy to see that the technique used in the kernel density estimation cannot be extended to the regression setup, since the commonly used kernel regression estimators, either the Nadaraya-Watson estimator, or the local linear estimator, often appear as a fraction of kernel components, which fails to provide a tractable conditional expectation for direct extrapolation. However, the observation of recovering the commonly used bias-corrected estimators or the SIMEX estimators derived using NON-IID pseudo-errors in the linear errors-in-variables regression indicates that we could have some interesting findings if we can apply the conditional expectation argument directly on the target functions, instead of computing the conditional expectation of the resulting naive estimator. In the next section, we will implement this idea via estimating the nonparametric regression function using a local linear smoothing procedure.

  \section{Extrapolation Estimation Procedure via Local Linear Smoother}\label{sec3}

  For the sake of simplicity, we restrict ourselves to the univariate predictor cases. The proposed methodology can handle the multivariate predictor cases very well at the cost of introducing more complex notations. To be specific, suppose that the random pair $(X,Y)$ obeys the following nonparametric regression model
     \begin{equation}\label{eq3.1}
        Y=g(X)+\vep,\quad Z=X+U
     \end{equation}
  with the common assumption on $\vep$, $E(\vep|X)=0$ and $0<\tau^2(X)=E(\vep^2|X)<\infty$. $X$ and $U$ are independent and $U$ has a normal distribution $N(0,\sigma_u^2)$ with known $\sigma_u^2$. If $(X,Y)$ are available, the local linear estimator for $g(x)$ at a fixed $x$-value in the domain of $X$ is defined as
      $$
        \hat g_n(x)=\frac{S_{2n}(x)T_{0n}(x)-S_{1n}(x)T_{1n}(x)}{S_{2n}(x)S_{0n}(x)-S_{1n}^2(x)},
      $$
  where
    $
      S_{jn}(x)=n^{-1}\sum_{i=1}^n(X_i-x)^jK_h(X_i-x), T_{jn}(x)=n^{-1}\sum_{i=1}^n(X_i-x)^jY_iK_h(X_i-x),
    $
  and $j=0,1,2$ for $S_{jn}(x)$, $j=0,1$ for $T_{jn}(x)$, $K_h(\cdot)=h^{-1}K(\cdot/h)$, and $K$ is a kernel function, $h$ is a sequence of positive numbers often called bandwidths. In the measurement error setup, a classical SIMEX estimator of $g$ can be obtained through three steps: simulation, estimation and extrapolation. For the sake of completeness, the following algorithm provides a detailed guideline for implementing the three steps in estimating $g(x)$ from data on $Y, Z$.
   \begin{center}
   \underline{~~~~~~~~~~~~~~~~~~~~~~~~~~~~~~~~~~~~~~~~~~~~~~~~~~~~~~~~~~~~~~~~~~~~~~~~~~~~~~~~~~~~~~~~~~~~~~~~~~~~~~~~~~~~~~~~~~}
   \vskip 0.1in
   \begin{minipage}[c]{5.5in}
      \hskip -0.3in {\bf SIMEX Algorithm of Local Linear Smoother}\vskip 0.1in
       \begin{enumerate}
    \item[(1)] Pre-select a sequence of positive numbers $\lambda=\lambda_1,\ldots,\lambda_K$.
    \item[(2)] For $\lambda=\lambda_1$, repeat the following steps $B$ times. At the $b$-th repetition,
      \begin{enumerate}
         \item[(i)] Generate $n$ i.i.d. random observations $V_{i,b}$'s from $N(0,\Sigma)$, and calculate
                    $Z_{i,b}(\lambda)=Z_i+\sqrt{\lambda_1}V_i=X_i+U_i+\sqrt{\lambda_1}V_{i,b}$, $i=1,2,\ldots,n.$
         \item[(ii)] Compute
               $$
                \hat g_{n,b}(x;\lambda_1)=\frac{S_{2nb}(x)T_{0nb}(x)-S_{1nb}(x)T_{1nb}(x)}{S_{2nb}(x)S_{0nb}(x)-S_{1nb}^2(x)},
               $$
         where
          \begin{eqarray*}
             S_{jnb}(x)&=&\frac{1}{n}\sum_{i=1}^n(Z_{i,b}(\lambda)-x)^jK_h(Z_{i,b}(\lambda)-x),\quad j=0,1,2,\\
             T_{lnb}(x)&=&\frac{1}{n}\sum_{i=1}^n(Z_{i,b}(\lambda)-x)^lY_iK_h(Z_{i,b}(\lambda)-x), \quad l=0,1.
          \end{eqarray*}\\ \vskip -0.3in \noindent
      \end{enumerate}
    \item[(3)] Calculate
                   $
                      \hat g_{n,B}(x;\lambda_1)=B^{-1}\sum_{b=1}^B\hat g_{n,b}(x;\lambda_1).
                   $
    \item[(4)] Repeat (2)-(3) for $\lambda=\lambda_2, \ldots, \lambda_K$.
    \item[(5)] Identify a parametric trend of the pairs $(\lambda_k, \hat g_{n,B}(x;\lambda_k))$, $k=1,2,\ldots, K$ and denote the trend as a function $\Gamma(x; \lambda)$. The SIMEX estimator of $g$ is defined as  $\hat g_{\tiny\rm SIMEX}(x)=\Gamma(x; -1)$.
  \end{enumerate}
   \end{minipage}\vskip 0.1in
   \underline{~~~~~~~~~~~~~~~~~~~~~~~~~~~~~~~~~~~~~~~~~~~~~~~~~~~~~~~~~~~~~~~~~~~~~~~~~~~~~~~~~~~~~~~~~~~~~~~~~~~~~~~~~~~~~~~~~~}
   \end{center}
 As a rough guideline, the $\lambda$ values are often selected as a sequence of equally spaced grid points from $[0,2]$,  $K$ is a positive integer as small as $5$ or as large as $20$, and $B$ is often chosen to be $100$ or above. With such choices, one can see the classical SIMEX procedure for implementing the local linear smoother is computationally intensive.

 To introduce our estimation procedure, we start with the local linear procedure. If $X$ can be observed, then based on a sample $(X_i,Y_i), i=1,2,\ldots,n$ from model (\ref{eq3.1}), the local linear estimator of the regression function $g$, as well as its first order derivative at $x$, can be obtained by minimizing the following target function
   $L(\beta_0,\beta_1)=\sum_{i=1}^n(Y_i-\beta_0-\beta_1(X_i-x))^2K_h(X_i-x)$
 with respect to $\beta_0$ and $\beta_1$. In fact, the solution of $\beta_0$ is the local linear estimator of $g(x)$ and $\beta_1$ is the local linear estimator of $g'(x)$.

 For a positive constant $\lambda$, we replace $X_i$ with the pesudo-data $Z_i(\lambda)=Z_i+\sqrt{\lambda}V_i$ in the weighted least squares $L(\beta_0,\beta_1)$, and calculate its conditional expectation given $(Z_i,Y_i), i=1,2,\ldots,n$. A straightforward calculation shows that the minimizer of
      $$
         \sum_{i=1}^n E\left([Y_i-\beta_0-\beta_1(Z_{i}(\lambda)-x)]^2K_h(x-Z_i(\lambda))|(Y_i,Z_i)\right)
      $$
 with respect to $\beta_0, \beta_1$ is given by the solution of the following equations
       \begin{equation}\label{eq3.2}
        \left\{
        \begin{array}{l}
         \displaystyle\sum_{i=1}^n E\left([Y_i-\beta_0-\beta_1(Z_{i}(\lambda)-x)]K_h(Z_i(\lambda)-x)|(Y_i,Z_i)\right)=0,\\[0.13in]
         \displaystyle\sum_{i=1}^n E\left([Y_i-\beta_0-\beta_1(Z_{i}(\lambda)-x)](Z_{i}(\lambda)-x)K_h(Z_i(\lambda)-x)|(Y_i,Z_i)\right)=0.
        \end{array}
        \right.
       \end{equation}
 The choice of kernel function $K$ is not critical in theory, but for the ease of computation, choosing $K$ to be standard normal can bring us extra benefits. In fact, with such a choice, together with the normality of the measurement error, the conditional expectations in (\ref{eq3.2}) have explicit forms. Note that $V_i$'s are i.i.d. from $N(0,\sigma_u^2)$ and independent of $(Z_i, Y_i)$, routine calculation (see Appendix) shows that
   \begin{eqarray}
     E[K_h(Z(\lambda)-x)|Y,Z]&=&\phi(x; Z, h^2+\lambda\sigma_u^2),\label{eq3.3}\\
     E[(Z(\lambda)-x)K_h(Z(\lambda)-x)|Y,Z]&=& \frac{h^2}{h^2+\lambda\sigma_u^2}(Z-x)\phi(x; Z,h^2+\lambda\sigma_u^2),\label{eq3.4}\\
     E[(Z(\lambda)-x)^2K_h(Z(\lambda)-x)|Y,Z]&=& \frac{h^4}{(h^2+\lambda\sigma_u^2)^2}(Z-x)^2\phi(x; Z,h^2+\lambda\sigma_u^2)\nonumber\\
       &&+\frac{\lambda\sigma_u^2h^2}{h^2+\lambda\sigma_u^2}\phi(x; Z,h^2+\lambda\sigma_u^2),\label{eq3.5}
   \end{eqarray}
 here, also throughout this paper, $\phi(x; \mu,\sigma_u^2)$ denotes the normal density function with mean $\mu$ and variance $\sigma_u^2$. Denote
  $ A_{nj}(x)=n^{-1}\sum_{i=1}^n(Z_i-x)^j\phi(x; Z_i, h^2+\lambda\sigma_u^2)$
 for $j=0,1,2$, and
  $B_{nl}(x)=n^{-1}\sum_{i=1}^nY_i(Z_i-x)^l\phi(x; Z_i, h^2+\lambda\sigma_u^2)$
 for $l=0,1$.
 Then the solution of $(\beta_0, \beta_1)$ of equation (\ref{eq3.2}), or $(\hat g_n(x;\lambda), \hat g_n'(x;\lambda))$ has the form of
  \begin{equation}\label{eq3.6}
    \begin{pmatrix}
      \hat g_n(x;\lambda) \\ \hat g_n'(x;\lambda)
    \end{pmatrix}=
    \begin{pmatrix}
      A_{n0}(x)  & r(\lambda,h)A_{n1}(x)\\
      r(\lambda,h)A_{n1}(x) &
      r(\lambda,h)[A_{n2}(x)+\lambda\sigma_u^2A_{n0}(x)]
    \end{pmatrix}^{-1}
    \begin{pmatrix}
      B_{n0}(x) \\  r(\lambda,h)B_{n1}(x)
    \end{pmatrix},
  \end{equation}
 where $r(\lambda,h)=h^2/(h^2+\lambda\sigma_u^2)$.

 Note that (\ref{eq3.6}) itself can be used for extrapolation. However, unlike the estimator $\hat\beta(\lambda)$, $\hat\alpha(\lambda)$ derived in the example of the linear regression,  $\lambda=-1$ cannot be plugged directly into (\ref{eq3.6}) to get the SIMEX estimator. In fact, when the sample size $n$ gets bigger, the bandwidth $h$ decreases to $0$. As a result, when $\lambda=-1$, $h^2+\lambda\sigma_u^2$ is negative for large sample sizes. As the variance of a normal density function, $h^2-\sigma_u^2$ should not be negative, which implies the extrapolation step is necessary.

 Therefore, we propose the following two-step SIMEX procedure, or more appropriately, the extrapolation (EX) procedure, to find an estimate of the regression function $g$. \vskip 0.15in
   \begin{center}
   \underline{~~~~~~~~~~~~~~~~~~~~~~~~~~~~~~~~~~~~~~~~~~~~~~~~~~~~~~~~~~~~~~~~~~~~~~~~~~~~~~~~~~~~~~~~~~~~~~~~~~~~~~~~~~~~~~~~~~}
   \vskip 0.1in
   \begin{minipage}[c]{5.5in}
      \hskip -0.3in {\bf EX Algorithm of The Local Linear Smoother}\vskip 0.1in
       \begin{enumerate}
        \item[(1)] For each $\lambda$ from the pre-selected sequence $\lambda=\lambda_1,\ldots,\lambda_K$, calculate $\hat g_n(x;\lambda)$;
        \item[(2)] Identify a trend of the pairs $(\lambda_k, \hat g_n(x;\lambda_k))$ and $(\lambda_k, \hat g_n'(x;\lambda_k))$, $k=1,2,\ldots, K$. Denote the trend as a function $G(x; \lambda)$, respectively. Then, the EX estimator of $g$ and its derivative are defined by $\hat g_{\tiny\rm EX}(x)=G(x; -1)$.
  \end{enumerate}
   \end{minipage}\vskip 0.1in
   \underline{~~~~~~~~~~~~~~~~~~~~~~~~~~~~~~~~~~~~~~~~~~~~~~~~~~~~~~~~~~~~~~~~~~~~~~~~~~~~~~~~~~~~~~~~~~~~~~~~~~~~~~~~~~~~~~~~~~}
   \end{center}

 Obviously, the above EX algorithm is much more efficient than the classical three-step SIMEX algorithm. Also,it is also easy to see that $\hat g_n(x;\lambda)$ from the EX algorithm is not the limit of $\hat g_{n,B}(x;\lambda)$ in the SIMEX algorithm as $B\to\infty$. Given the observed data $(Z_i,Y_i)_{i=1}^n$, by the law of large numbers, for a fixed $\lambda$, as $B\to\infty$, $\hat g_{n,B}(x;\lambda)=B^{-1}\sum_{b=1}^B\hat g_{n,b}(x;\lambda)$ converges to $\tilde g_n(x;\lambda)$ in probability, where
   \begin{eqarray}\label{eq3.7}
     \tilde g_n(x;\lambda)=\int\frac{S_{2n}(x,\bv)T_{0n}(x,\bv)-S_{1n}(x,\bv)T_{1n}(x,\bv)}{S_{2n}(x,\bv)S_{0n}(x,\bv)-S_{1n}^2(x,\bv)}
     \bphi(\bv;0,\lambda\sigma_u^2)d\bv,
   \end{eqarray}
 where $\bv=(v_1,\ldots, v_n)^T$, $\bphi(\bv;0,\lambda\sigma_u^2)=\prod_{i=1}^n\phi(v_i;0,\lambda\sigma_u^2)$, and
          \begin{eqarray*}
             S_{jn}(x,\bv)&=&\frac{1}{n}\sum_{i=1}^n(Z_i+v_i-x)^jK_h(Z_i+v_i-x),\quad j=0,1,2,\\
             T_{ln}(x,\bv)&=&\frac{1}{n}\sum_{i=1}^n(Z_i+v_i-x)^lY_iK_h(Z_i+v_i-x), \quad l=0,1.
          \end{eqarray*}\\ \vskip -0.3in \noindent
 However, the estimator $\hat g_n(x;\lambda)$ defined in (\ref{eq3.6}) has the form of
    \begin{eqarray}\label{eq3.8}
     \hat g_n(x;\lambda)&=&\frac{\tilde S_{2n}(x)\tilde T_{0n}(x)-\tilde S_{1n}(x)\tilde T_{1n}(x)}{\tilde S_{2n}(x)\tilde S_{0n}(x)-\tilde S_{1n}^2(x)}\nonumber\\
     &=&\frac{A_{n2}(x)B_{n0}(x)+\lambda\sigma_u^2r^{-1}(\lambda,h)A_{n0}(x)B_{n0}(x)-A_{n1}(x)B_{n1}(x)}
     {A_{n2}(x)A_{n0}(x)+\lambda\sigma_u^2r^{-1}(\lambda,h)A_{n0}^2(x)-A_{n1}(x)},
    \end{eqarray}
 where
   \begin{equation}\label{eq3.9}
      \tilde S_{jn}(x)=\int S_{jn}(x,\bv)\bphi(\bv;0,\lambda\sigma_u^2)d\bv, \quad \tilde T_{ln}(x)=\int T_{ln}(x,\bv)\bphi(\bv;0,\lambda\sigma_u^2)d\bv
   \end{equation}
 for $j=0,1,2$ and $l=0,1$, respectively. Therefore, $\tilde g_n(x;\lambda)$ is different from $\hat g_n(x;\lambda)$, which indicates that $\hat g_n(x;\lambda)$ from the SIMEX algorithm is not the limit of the EX algorithm  as $B\to\infty$. In fact, $\hat g_n(x;\lambda)$ can be viewed as the limit of $\hat g_{n,b}(x;\lambda)$ with $S_{jnb}(x)$ and $T_{lnb}(x)$ replaced by $B^{-1}\sum_{b=1}^BS_{jnb}(x)$ and $B^{-1}\sum_{b=1}^BT_{lnb}(x)$, $j=0,1,2$, $l=0,1$, respectively, as $B\to\infty.$

 \section{Asymptotic Theory of EX Algorithm}\label{sec4}

  In this section, we shall investigate the large sample behaviours for the EX algorithm proposed in the previous section.
  We will show that as $n\to\infty$, $\hat g_n(x;\lambda)$ indeed converges to a function of both $x$ and $\lambda$, but the latter can approximate the true regression function $g(x)$ as $\lambda\to -1$, thus justifying the effectiveness of extrapolation. The asymptotic joint distribution of $\hat g_n(x;\lambda)$ at different $\lambda$ values, including $\lambda=0$ which corresponds to the naive estimator, will be also discussed.

  The following is a list of regularity conditions we need to justify all the theoretical derivations.
   \begin{itemize}
     \item[] C1. $f_x(x)$, $g(x)$ $\tau^2(x)=E(\vep^2|X=x)$, $\mu(x)=E(|\vep|^3|X=x)$ are twice continuously differentiable; also for each $x$ in the support of $X$, as a function of $t$, $\eta'(t+x), \eta''(t+x)\in L_2(\phi(t,0,\sigma_u^2))$, where $\eta=f_x, g, g^2, \tau^2, \tau^4$, $\mu$ and $\mu^2$.
     \item[] C2. The bandwidth $h$ satisfies $h\to 0$, $nh\to\infty$ as $n\to\infty$.
   \end{itemize}

  To proceed, for integers $j\geq 0$, we denote
   \begin{eqarray*}
    f_{j,\lambda}(x)&=&\int \phi(t;x,(\lambda+1)\sigma_u^2)t^jf_X(t)dt,\quad g_{j,\lambda}(x)=\int t^jg(t)f_X(t)\phi(t,x,(1+\lambda)\sigma_u^2)dt,\\
    G_{j,\lambda}(x)&=&\int t^jg^2(t)f_X(t)\phi(t,x,(1+\lambda)\sigma_u^2)dt,\quad
    H_{j,\lambda}(x)=\int t^j\tau^2(t)f_X(t)\phi(t,x,(1+\lambda)\sigma_u^2)dt.
  \end{eqarray*}\\ \vskip -0.3in \noindent

  By a routine and tedious calculation, we can show the following result from which the asymptotic bias of $\hat g_n(x; \lambda)$ can be derived as $n\to\infty$.

  \begin{thm}\label{thm0}
    Under conditions C1 and C2, for each $\lambda\geq 0$, we have
  \begin{eqarray}\label{eq4.1}
    &&\frac{E\tilde S_{n2}(x)\cdot E\tilde T_{n0}(x)-E\tilde S_{n1}(x)\cdot E\tilde T_{n1}(x)}{E\tilde S_{n2}(x)\cdot E\tilde S_{n0}(x)-[E\tilde S_{n1}(x)]^2}=\frac{g_{0,\lambda}(x)}{f_{0,\lambda}(x)}+h^2B(x;\lambda)+o(h^2),
  \end{eqarray}
  where $B(x;\lambda)$ equals
  $$
   \frac{f_{0,\lambda}(x)g''_{0,\lambda}(x)-f_{0,\lambda}''(x)g_{0,\lambda}(x)}{2f_{0,\lambda}^2(x)}
   +\frac{(f_{1,\lambda}(x)-xf_{0,\lambda}(x))(g_{0,\lambda}(x)f_{1,\lambda}(x)-f_{0,\lambda}(x)g_{1,\lambda}(x))}
   {(\lambda+1)^2\sigma_u^4f_{0,\lambda}^3(x)},
  $$
  where $\tilde S_{nj}(x)$ and $\tilde T_{nl}(x)$ for $j=0,1,2$ and $l=0,1$ are defined in (\ref{eq3.9}).
  \end{thm}

 Note that, as $\lambda\to -1$, $g_{0,\lambda}(x)\to g(x)f_X(x)$, and $f_{0,\lambda}(x)\to f_X(x)$.
 For $B(x;\lambda)$, we have
   \begin{eqarray*}
     f_{1,\lambda}(x)-xf_{0,\lambda}(x)&=&\int (t-x)f_X(t)\phi(t; x,(\lambda+1)\sigma_u^2)dt=(\lambda+1)\sigma_u^2f_X'(x)+o((\lambda+1)\sigma_u^2),
   \end{eqarray*}\\ \vskip -0.3in \noindent
and $g_{1,\lambda}(x)-xg_{0,\lambda}(x)$ can be written as
   \begin{eqarray*}
    \int (t-x)g(t)f_X(t)\phi(t; x,(\lambda+1)\sigma_u^2)dt
    =(\lambda+1)\sigma_u^2[gf_X]'(x)+o((\lambda+1)\sigma_u^2)
   \end{eqarray*}\\ \vskip -0.3in \noindent
as $\lambda\to -1$. Then we can further show that
  \begin{eqarray}\label{eq4.2}
   \hskip -0.3in B(x;\lambda)
   =\frac{f_{0,\lambda}(x)g''_{0,\lambda}(x)-f_{0,\lambda}''(x)g_{0,\lambda}(x)}{2f_{0,\lambda}^2(x)}
   +\frac{g_{0,\lambda}(x)(f_X'(x))^2}{f_{0,\lambda}^3(x)}-\frac{f_X'(x)[gf_X]'(x)}{f_{0,\lambda}^2(x)}
   +o(1),
  \end{eqarray}\\ \vskip -0.3in \noindent
 where $o(1)$ denotes that the corresponding terms converge to $0$ as $\lambda\to -1$. Therefore, we have
  $
    \lim_{\lambda\to -1}B(x;\lambda)=g''(x)/2.
  $
 Thus, from Theorem \ref{thm0},
   $$
    \lim_{\lambda\to -1}\left[\frac{E\tilde S_{n2}(x)\cdot E\tilde T_{n0}(x)-E\tilde S_{n1}(x)\cdot E\tilde T_{n1}(x)}{E\tilde S_{n2}(x)\cdot E\tilde S_{n0}(x)-[E\tilde S_{n1}(x)]^2}\right]=g(x)+\frac{g''(x)h^2}{2}+o(h^2),
   $$
 and this immediately leads to
  $$
   \lim_{\lambda\to -1}\lim_{h\to 0}\left[\frac{E\tilde S_{n2}(x)\cdot E\tilde T_{n0}(x)-E\tilde S_{n1}(x)\cdot E\tilde T_{n1}(x)}{E\tilde S_{n2}(x)\cdot E\tilde S_{n0}(x)-[E\tilde S_{n1}(x)]^2}\right]=g(x).
  $$

 To investigate the asymptotic distribution of $\hat g_n(x;\lambda)$, denote
  \begin{eqarray*}
    D_{n}(x)&=&(\tilde S_{n2}(x)\tilde S_{n0}(x)-\tilde S_{n1}^2(x))(E\tilde S_{n2}(x)E\tilde S_{n0}(x)-(E\tilde S_{n1}(x))^2),\\
    C_{n0}(x)&=&E\tilde S_{n2}(x)[E\tilde S_{n1}(x)E\tilde T_{n1}(x)-E\tilde S_{n2}(x)E\tilde T_{n0}(x)],\\
    C_{n1}(x)&=&2E\tilde S_{n1}(x)E\tilde S_{n2}(x)E\tilde T_{n0}(x)-(E\tilde S_{n1}(x))^2E\tilde T_{n1}(x)-E\tilde T_{n1}(x)E\tilde S_{n2}(x)E\tilde S_{n0}(x),\\
    C_{n2}(x)&=&E\tilde S_{n1}(x)\cdot[E\tilde S_{n0}(x)E\tilde T_{n1}(x)-E\tilde T_{n0}(x)E\tilde S_{n1}(x)],\\
    D_{n0}(x)&=&E\tilde S_{n2}(x)[E\tilde S_{n2}(x)E\tilde S_{n0}(x)-(E\tilde S_{n1}(x))^2],\\
    D_{n1}(x)&=&E\tilde S_{n1}(x)[(E\tilde S_{n1}(x))^2-E\tilde S_{n2}(x)E\tilde S_{n0}(x)].
  \end{eqarray*}\\ \vskip -0.5in \noindent
 From Theorem \ref{thm0}, we can write $\hat g_n(x)$ as
  \begin{eqarray}\label{eq4.3}
     \hat g_n(x;\lambda)&=&\frac{g_{0,\lambda}(x)}{f_{0,\lambda}(x)}+h^2B(x;\lambda)+o(h^2)\nonumber\\
       &&+D_{n}^{-1}(x)\Big[
     \sum_{j=0}^2C_{nj}(x)(\tilde S_{nj}-E\tilde S_{nj})+
      \sum_{l=0}^1D_{nl}(x)(\tilde T_{nl}-E\tilde T_{nl})\Big].
  \end{eqarray}\\ \vskip -0.4in \noindent
 Denote
   \begin{eqarray*}
     c_{0\lambda}(x)&=&-\frac{g_{0,\lambda}(x)}{f_{0,\lambda}^2(x)},\quad c_{1\lambda}(x)=\frac{2[f_{1,\lambda}(x)-xf_{0,\lambda}(x)]g_{0,\lambda}(x)-
        [g_{1,\lambda}(x)-xg_{0,\lambda}(x)]f_{0,\lambda}(x)}{(\lambda+1)\sigma_u^2f_{0,\lambda}^3(x)},\\
     c_{2\lambda}(x)&=& \frac{[f_{1,\lambda}(x)-xf_{0,\lambda}(x)][g_{1,\lambda}(x)-xg_{0,\lambda}(x)]f_{0,\lambda}(x)-
        [f_{1,\lambda}(x)-xf_{0,\lambda}(x)]^2g_{0,\lambda}(x)}{(\lambda+1)^2\sigma_u^4f_{0,\lambda}^4(x)},\\
     d_{0\lambda}(x)&=&\frac{1}{f_{0,\lambda}(x)}, \quad d_{1\lambda}(x)=-\frac{f_{1,\lambda}(x)-xf_{0,\lambda}(x)}{(\lambda+1)\sigma_u^2f_{0,\lambda}^2(x)}.
   \end{eqarray*}\\ \vskip -0.3in \noindent
 Then, from Lemma \ref{lem3} - Lemma \ref{lem7} in Appendix, we can show that, for $\lambda\geq 0$, $     D_{n}^{-1}C_{n}=c_{j\lambda}(x)+o_p(1)$, for $j=0,1,2$ and $D_{n}^{-1}D_{nj}=d_{j\lambda}(x)+o_p(1)$
 for $j=0,1$. We further denote
  \begin{eqarray*}
    \xi_{0\lambda,i}(x)&=&\phi(x,Z_i,h^2+\lambda\sigma_u^2)-E\phi(x,Z,h^2+\lambda\sigma_u^2),\\
    \xi_{1\lambda,i}(x)&=& \frac{h^2}{h^2+\lambda\sigma_u^2}\left[(Z_i-x)\phi(x,Z_i,h^2+\lambda\sigma_u^2)-E(Z-x)\phi(x,Z,h^2+\lambda\sigma_u^2)\right],\\
    \xi_{2\lambda,i}(x)&=&\frac{h^4}{(h^2+\lambda\sigma_u^2)^2}\left[(Z_i-x)^2\phi(x,Z_i,h^2+\lambda\sigma_u^2)
    -E(Z-x)^2\phi(x,Z,h^2+\lambda\sigma_u^2)\right]\\
    &&+\frac{\lambda\sigma_u^2h^2}{h^2+\lambda\sigma_u^2}\left[\phi(x,Z_i,h^2+\lambda\sigma_u^2)-E\phi(x,,h^2+\lambda\sigma_u^2)\right],\\
    \eta_{0\lambda,i}(x)&=& Y_i\phi(x,Z_i,h^2+\lambda\sigma_u^2)-EY\phi(x,Z,h^2+\lambda\sigma_u^2)\\
    \eta_{1\lambda,i}(x)&=& \frac{h^2}{h^2+\lambda\sigma_u^2}\left[Y_i(Z_i-x)\phi(x,Z_i,h^2+\lambda\sigma_u^2)-EY(Z-x)\phi(x,Z,h^2+\lambda\sigma_u^2)\right].
  \end{eqarray*}\\ \vskip -0.3in \noindent
 Then, from (\ref{eq4.3}), we have
   \begin{eqarray*}
    \hat g_n(x;\lambda)&=&\frac{g_{0,\lambda}(x)}{f_{0,\lambda}(x)}+h^2B(x;\lambda)+o(h^2)\\
       &&\sum_{j=0}^2 [c_{j\lambda}(x)+o(1)](\tilde S_{nj}-E\tilde S_{nj})+\sum_{k=0}^1[d_{k\lambda}(x)+o(1)](\tilde T_{nk}-E\tilde T_{nk}).
   \end{eqarray*}\\ \vskip -0.3in \noindent
 Since the terms $o(1)$ in the above expression does not affect the asymptotic distribution of $\hat g_n(x;\lambda)$, so we can safely neglect the $o(1)$ terms from the sum, and therefore the two sums can be written as an i.i.d. average $n^{-1}\sum_{i=1}^nv_{i\lambda}(x)$, where $v_{i\lambda}(x)$ is defined by
  \begin{equation}\label{eq4.4}
    c_{0\lambda}(x)\xi_{0\lambda,i}(x)+
     c_{1\lambda}(x)\xi_{1\lambda,i}(x)+c_{2\lambda}(x)\xi_{2\lambda,i}(x)+
     d_{0\lambda}(x)\eta_{0\lambda,i}(x)+d_{1\lambda}(x)\eta_{1\lambda,i}(x).
  \end{equation}\\ \vskip -0.4in \noindent
 By verifying the Lyapunov condition, we can show that for each $\lambda>0$, $n^{-1}\sum_{i=1}^nv_{i\lambda}(x)$ is asymptotically normal. This asymptotic normality is summarized in the following theorem.

 \begin{thm}\label{thm1} Under conditions C1 and C2, for each $\lambda> 0$,
   $$
     \sqrt{n}\left\{ \hat g_n(x;\lambda)-\frac{g_{0,\lambda}(x)}{f_{0,\lambda}(x)}-h^2B(x;\lambda)+o(h^2)
     \right\}\Longrightarrow N(0, \Delta_{\lambda,\lambda}(x)),
   $$
 and for $\lambda=0$,
   $$
     \sqrt{nh}\left\{ \hat g_n(x;0)-\frac{g_{0,0}(x)}{f_{0,0}(x)}-h^2B(x;0)+o(h^2)
     \right\}\Longrightarrow N(0, \Delta_{0,0}(x)),
   $$
 where
  \begin{eqarray*}
    \Delta_{\lambda,\lambda}(x)&=&c_{0\lambda}^2\left[\frac{ f_{0,\lambda/2}(x)}{2\sqrt{\pi\lambda\sigma_u^2}}-f_{0,\lambda}^2(x)\right]
       +d_{0\lambda}^2\left[\frac{G_{0,\lambda/2}(x)+H_{0,\lambda/2}(x)}{2\sqrt{\pi\lambda\sigma_u^2}}-g_{0,\lambda}^2\right]\\
           && +2c_{0\lambda}d_{0\lambda}\left[\frac{ g_{0,\lambda/2}(x)}{2\sqrt{\pi\lambda\sigma_u^2}}-g_{0,\lambda}(x)f_{0,\lambda}(x)\right]
  \end{eqarray*}\\ \vskip -0.3in \noindent
 and
  \begin{eqarray*}
    \Delta_{0,0}(x)&=&
    \frac{1}{2\sqrt{\pi}}\left[\frac{G_{00}(x)+H_{00}(x)}{f_{00}^2(x)}-\frac{g_{00}^2(x)}{f_{00}^3(x)}\right].
  \end{eqarray*}\\ \vskip -0.3in \noindent
\end{thm}

Note that when $\sigma_u^2=0$, that is, no measurement error in $X$, then one can easily see that $\Delta_{0,0}(x)=\tau^2(x)/(2\sqrt{\pi}f_X(x))$, which is exactly the asymptotic variance in local linear estimator of the regression function in the error-free cases.
The theorem below states the asymptotic joint normality of $[\hat g_n(x;0), \hat g_n(x;\lambda_1), \cdots,\hat g_n(x;\lambda_K)]'$.

 \begin{thm}\label{thm2} Under conditions C1 and C2, for $0<\lambda_1<\cdots<\lambda_K<\infty$,
   $$
     \begin{pmatrix}
        \sqrt{nh} & 0 & \cdots & 0 \\
           0     & \sqrt{n}& \cdots & 0 \\
         \vdots  &  \vdots  & \ddots & \vdots\\
           0     &  0       &  \cdots & \sqrt{n}
     \end{pmatrix}
     \begin{pmatrix*}[l]
       \hat g_n(x;0)-g_{0,0}(x)/f_{0,0}(x)-h^2B(x;0)+o(h^2)\\
       \hat g_n(x;\lambda_1)-g_{0,\lambda_1}(x)/f_{0,\lambda_1}(x)-h^2B(x;\lambda_1)+o(h^2)\\
       \vdots \\
       \hat g_n(x;\lambda_K)-g_{0,\lambda_K}(x)/f_{0,\lambda_K}(x)-h^2B(x;\lambda_K)+o(h^2)
     \end{pmatrix*}
   $$
   $$
    \mbox{}\hskip -4in  \Longrightarrow N(0, \Delta(x)),
   $$
 where $B(x;\lambda)$ is defined in (\ref{eq4.2}),
   $$
   \Delta(x)=
    \begin{pmatrix}
      \Delta_{0,0}(x)  &  0  &  0  &  \cdots & 0\\
        0 & \Delta_{\lambda_1,\lambda_1}(x)  & \Delta_{\lambda_1\lambda_2}(x) & \cdots & \Delta_{\lambda_1\lambda_K}(x)\\
        0 & \Delta_{\lambda_1\lambda_2}(x) & \Delta_{\lambda_2,\lambda_2}(x) & \cdots & \Delta_{\lambda_2\lambda_K}(x)\\
        \vdots & \Delta_{\lambda_1\lambda_K}(x)  & \Delta_{\lambda_2\lambda_K}(x) & \cdots & \Delta_{\lambda_K,\lambda_K}(x)
    \end{pmatrix},
  $$
 and $\Delta_{\lambda_i\lambda_j}(x)$, $i=0, 1,\ldots,K$, $j=1,2,\ldots,K$, are given by
    \begin{eqarray*}
      &&\frac{c_{0\lambda_i}(x)c_{0\lambda_j}(x)}{\sqrt{2\pi(\lambda_i+\lambda_j)\sigma_u^2}}\int \phi\left(t, x, \left(\frac{\lambda_i\lambda_j}{\lambda_i+\lambda_j}+1\right)\sigma_u^2\right)f_X(t)dt-f_{0,\lambda_i}(x)f_{0,\lambda_j}(x)\\
      &&+\frac{c_{0\lambda_i}(x)d_{0\lambda_j}(x)}{\sqrt{2\pi(\lambda_i+\lambda_j)\sigma_u^2}}\int g(t)\phi\left(t, x, \left(\frac{\lambda_i\lambda_j}{\lambda_i+\lambda_j}+1\right)\sigma_u^2\right)f_X(t)dt
      -f_{0,\lambda_i}(x)g_{0,\lambda_j}(x)\\
      &&+\frac{c_{0\lambda_j}(x)d_{0\lambda_i}(x)}{\sqrt{2\pi(\lambda_i+\lambda_j)\sigma_u^2}}\int g(t)\phi\left(t, x, \left(\frac{\lambda_i\lambda_j}{\lambda_i+\lambda_j}+1\right)\sigma_u^2\right)f_X(t)dt
      -f_{0,\lambda_j}(x)g_{0,\lambda_i}(x)\\
      &&+\frac{d_{0\lambda_i}(x)d_{0\lambda_j}(x)}{\sqrt{2\pi(\lambda_i+\lambda_j)\sigma_u^2}}\int g^2(t)\phi\left(t, x, \left(\frac{\lambda_i\lambda_j}{\lambda_i+\lambda_j}+1\right)\sigma_u^2\right)f_X(t)dt
      -g_{0,\lambda_i}(x)g_{0,\lambda_j}(x).
    \end{eqarray*}\\ \vskip -0.3in \noindent
\end{thm}
 The proof of the joint normality is a straightforward application of the multivariate CLT on the following random vector
   $$
     \begin{pmatrix*}[c]
       \sqrt{nh}\left[\hat g_n(x;0)-g_{0,0}(x)/f_{0,0}(x)-h^2B(x;0)+o(h^2)\right]\\
       \sqrt{n}\left[\hat g_n(x;\lambda_1)-g_{0,\lambda_1}(x)/f_{0,\lambda_1}(x)-h^2B(x;\lambda_1)+o(h^2)\right]\\
        \vdots \\
       \sqrt{n}\left[\hat g_n(x;\lambda_K)-g_{0,\lambda_K}(x)/f_{0,\lambda_K}(x)-h^2B(x;\lambda_K)+o(h^2)\right]
     \end{pmatrix*}=\frac{1}{\sqrt{n}}\begin{pmatrix*}[l]
     \sqrt{h}\sum_{i=1}^n v_{i0}(x)\\
     \sum_{i=1}^n v_{i\lambda_1}(x)\\
     \hfill\vdots\hfill\\
     \sum_{i=1}^n v_{i\lambda_K}(x)
   \end{pmatrix*}.
   $$
 For the sake of brevity, the proof will be omitted. In addition to the condition C2, if we further assume that $nh^4\to 0$, then the asymptotic bias can be removed.

\section{Extrapolation Function}\label{sec5}

 From the discussion in the previous section, the extrapolation function can be derived from $g_{0,\lambda}(x)/f_{0,\lambda}(x)$. From the definitions of $g_{0,\lambda}$ and $f_{0,\lambda}(x)$, we know that
   \begin{equation}\label{eq5.1}
     \Gamma(\lambda):=\frac{g_{0,\lambda}(x)}{f_{0,\lambda}(x)}=\frac{\int g(t)f_X(t)\phi(t; x,(\lambda+1)\sigma_u^2)dt}{\int f_X(t)\phi(t; x,(\lambda+1)\sigma_u^2)dt}.
   \end{equation}
 As a function of $\lambda$, $\Gamma(\lambda)$ does not have a tractable form, and some approximation is needed for extrapolating.
 By change of variable, we have
   $$
    \int g(t)f_X(t)\phi(t; x,(\lambda+1)\sigma_u^2)dt
    =\int g(x+\sqrt{\lambda+1}\sigma_uv)f_X(x+\sqrt{\lambda+1}\sigma_uv)\phi(v)dv.
   $$
 Denote $\alpha=(\lambda+1)\sigma_u^2$, and assume that $g$ and $f_X$ are four times continuously differentiable. Then we have
   $$
    \int g(t)f_X(t)\phi(t; x,(\lambda+1)\sigma_u^2)dt=g(x)f_X(x)+\frac{[f_X(x)g(x)]''}{2}\alpha+
    \frac{[f_X(x)g(x)]^{(4)}}{4!}\alpha^2+o(\alpha^2),
   $$
 where $o(\cdot)$ is understood as a negligible quantity when $\lambda\to -1$.
 Similarly, we have
   $$
    \int f_X(t)\phi(x; t,(\lambda+1)\sigma_u^2)dt=f_X(x)+\frac{f_X''(x)}{2}\alpha+\frac{f_X^{(4)}(x)}{4!}\alpha^2+o(\alpha^2).
   $$
 Therefore, after neglecting the $o(\lambda+1)$ term, from (\ref{eq5.1}), we obtain
  $$
    \Gamma(\lambda)\approx \frac{g(x)f_X(x)+\sigma_u^2(\lambda+1)[f_X(x)g(x)]''/2}{f_X(x)+\sigma_u^2(\lambda+1)f_X''(x)/2}.
  $$
 It is easy to see that the right hand side approaches $g(x)$ as $\lambda\to -1$, and indeed, for fixed $x$-value, it has the form of $a+b/(c+\lambda)$, the nonlinear extrapolation function often used in the classical SIMEX estimation procedure. If we further apply the approximation
   $$
      \frac{1}{f_X(x)+\sigma_u^2(\lambda+1)f_X''(x)/2}
      =\frac{1}{f_X(x)}\left[1-\frac{\sigma_u^2(\lambda+1)f_X''(x)}{2f_X(x)}+o((\lambda+1))\right]
   $$
 or the approximation with higher order expansions, then we can obtain the commonly used quadratic extrapolation function $a+b\lambda+c\lambda^2$ and the polynomial extrapolation functions.

 Almost all literature involving the classical SIMEX method, mostly in the parametric setups, assumes that the true extrapolation function has a known nonlinear form when discussing the asymptotic distributions of the SIMEX estimators. However, based on the above discussion, the true extrapolation function is never known. To see this point clearly, we further assume that $X\sim N(0,\sigma_x^2)$. Then from (\ref{eq5.1}), for any $x\in\R$,
   $$
      \Gamma(\lambda)=\int g(t)\phi\left(t,\frac{x\sigma_x^2}{(\lambda+1)\sigma_u^2+\sigma_x^2},
      \frac{(\lambda+1)\sigma_u^2\sigma_x^2}{(\lambda+1)\sigma_u^2+\sigma_x^2}
      \right)dt.
   $$
 Since the normal distribution family is complete, so the above expression implies that $\Gamma(\lambda)$ and $g(t)$ are uniquely determined by each other. Since $g$ is unknown, so neither is $\Gamma$. This discouraging finding really invalidates all the potential theoretical developments based on known extrapolation functions.

\section{Numerical Study}\label{sec6}

 In this section, we conduct some simulation studies to evaluate the finite sample performance of the proposed SIMEX procedure. We also analyze a dataset from the National Health and Nutrition Examination Survey (NHANES) to illustrate the application of the proposed estimation procedure.  \vskip 0.1in

\subsection{Simulation Study}

  In this simulation study, the simulated data are generated from the regression model $Y=g(X)+\vep, Z=X+U$, where $X$ has a standard normal distribution, $U$ is generated from $N(0,\sigma_u^2)$. Three choices of regression function $g(x)$ were considered, namely $g(x)=x^2$, $\exp(x)$ and $x\sin(x)$. To see the effect of the measurement error variance on the resulting estimate, we choose $\sigma_u^2=0.1$ and $0.25$. The sample sizes are chosen to be $n=100, 200, 500$. In each scenario, the estimates are calculated for $200$ equally spaced $x$-values $x_j$, $j=1,2,\ldots, 200$, are chosen from $[-3,3]$. To implement the extrapolation step, the grid of $\lambda$ is taken from 0 to 2 separated by 0.2. The mean squared errors (MSE) are used to evaluate the finite sample performance of the proposed SIMEX procedure. The bandwidth $h$ is chosen to be $n^{-1/5}$, a theoretical optimal order when estimating the regression function based on the error-free data. To get a stable result, all simulations were performed for 10 independent datasets, and the average of the 10 estimates was taken to be the final estimate at each of 200 $x$-values, and the MSE defined by $200^{-1}\sum_{j=1}^{200}[\hat g_n(x_j)-g(x_j)]^2$ is used for evaluate the finite sample performance of the proposed EX estimate. For comparison, we also apply the classical SIMEX algorithm and the naive method to estimate these three regression functions with $B=50, 100$. Besides the report on the MSEs from the three algorithms, we also record the computation time in seconds from each procedure to evaluate the algorithm efficiency. The simulation results are summarized in Tables \ref{tab1} - \ref{tab3}. In all the tables, we use the EX to denote the proposed Extrapolation algorithm, SIMEX to denote the classical SIMEX method, and Naive for naive method.
	
 \begin{table}[h!]
  \begin{center}
    \caption{$g(x)=x\sin(x)$, $X\sim N(0,1)$}
    \label{tab1}
    \begin{tabular}{lclrrrrrr}
      \toprule 
      \textbf{$\sigma_u^2$} & Method &  & \multicolumn{2}{c}{\textbf{$n=100$}} &  \multicolumn{2}{c}{\textbf{$n=200$}} & \multicolumn{2}{c}{\textbf{$n=500$}}\\
      \cmidrule(lr){4-5} \cmidrule(lr){6-7} \cmidrule(lr){8-9}
      &     & & MSE & Time(s) & MSE & Time(s) & MSE & Time(s) \\
      \midrule 
          & SIMEX & $B=50$  & 0.066 &  71.001 & 0.065 & 122.568 & 0.023 & 285.839\\
      \Lower{0.1}&& $B=100$ & 0.142 & 139.862 & 0.137 & 243.531 & 0.032 & 569.185\\
          & EX    &         & 0.367 &   2.264 & 0.079 &   3.120 & 0.051 &   5.853\\
          & Naive &         & 0.197 &   0.180 & 0.101 &   0.274 & 0.066 &   0.569\\      \midrule
          & SIMEX & $B=50$  & 0.075 &  71.053 & 0.047 & 123.463 & 0.077 & 287.537\\
     \Lower{0.25}&& $B=100$ & 0.084 & 141.015 & 0.138 & 245.227 & 0.067 & 572.512\\
          & EX    &         & 0.335 &   2.202 & 0.118 &   2.948 & 0.065 &   5.334\\
          & Naive &         & 0.221 &   0.176 & 0.139 &   0.273 & 0.119 &   0.574\\
      \bottomrule 
    \end{tabular}
  \end{center}
\end{table}

 \begin{table}[h!]
  \begin{center}
    \caption{$g(x)=x^2$, $X\sim N(0,1)$}
    \label{tab2}
    \begin{tabular}{lclrrrrrr}
      \toprule 
      \textbf{$\sigma_u^2$} & Method &  & \multicolumn{2}{c}{\textbf{$n=100$}} &  \multicolumn{2}{c}{\textbf{$n=200$}} & \multicolumn{2}{c}{\textbf{$n=500$}}\\
      \cmidrule(lr){4-5} \cmidrule(lr){6-7} \cmidrule(lr){8-9}
      &     & & MSE & Time(s) & MSE & Time(s) & MSE & Time(s) \\
      \midrule 
          & SIMEX & $B=50$  & 0.420 &  70.813 & 0.346 & 122.444 & 0.055 & 285.444\\
      \Lower{0.1}&& $B=100$ & 0.165 & 139.648 & 0.192 & 243.118 & 0.154 & 568.681\\
          & EX    &         & 0.688 &   2.263 & 0.201 &   3.129 & 0.067 &   5.884\\
          & Naive &         & 0.318 &   0.179 & 0.450 &   0.270 & 0.385 &   0.572\\      \midrule
          & SIMEX & $B=50$  & 0.996 &  71.103 & 0.336 & 123.314 & 0.164 & 287.267\\
     \Lower{0.25}&& $B=100$ & 0.711 & 140.884 & 0.070 & 245.532 & 0.196 & 573.299\\
          & EX    &         & 0.069 &   2.176 & 0.372 &   2.949 & 0.029 &   5.340\\
          & Naive &         & 1.820 &   0.180 & 2.282 &   0.274 & 1.456 &   0.573\\
      \bottomrule 
    \end{tabular}
  \end{center}
\end{table}

 \begin{table}[h!]
  \begin{center}
    \caption{$g(x)=\exp(x)$, $X\sim N(0,1)$}
    \label{tab3}
    \begin{tabular}{lclrrrrrr}
      \toprule 
      \textbf{$\sigma_u^2$} & Method &  & \multicolumn{2}{c}{\textbf{$n=100$}} &  \multicolumn{2}{c}{\textbf{$n=200$}} & \multicolumn{2}{c}{\textbf{$n=500$}}\\
      \cmidrule(lr){4-5} \cmidrule(lr){6-7} \cmidrule(lr){8-9}
      &     & & MSE & Time(s) & MSE & Time(s) & MSE & Time(s) \\
      \midrule 
          & SIMEX & $B=50$  & 1.542 &  70.906 & 1.654 & 122.614 & 0.060 & 285.777\\
      \Lower{0.1}&& $B=100$ & 0.351 & 139.888 & 0.069 & 243.476 & 0.252 & 568.903\\
          & EX    &         & 1.144 &   2.250 & 0.125 &   3.157 & 0.149 &   5.826\\
          & Naive &         & 0.637 &   0.175 & 1.282 &   0.273 & 0.798 &   0.557\\ \midrule
          & SIMEX & $B=50$  & 2.241 &  71.761 & 1.707 & 123.365 & 0.898 & 287.329\\
     \Lower{0.25}&& $B=100$ & 0.544 & 140.633 & 0.150 & 245.134 & 0.217 & 572.378\\
          & EX    &         & 0.208 &   2.176 & 0.436 &   2.950 & 0.070 &   5.342\\
          & Naive &         & 3.772 &   0.176 & 3.785 &   0.274 & 2.749 &   0.570\\
      \bottomrule 
    \end{tabular}
  \end{center}
\end{table}
  The simulation results clearly show that the proposed EX algorithm is more efficient than the classical SIMEX method in terms of computational speed. The finite sample performance of both methods SIMEX and EX, as measured by the MSE, are comparable. When sample sizes get larger, and the measurement error variances get smaller, both procedures performs better, as expected. The advantage of using EX or SIMEX over the naive method may not be obvious when the sample size or the noise level $\sigma_u^2$ is small, but both methods outperform the naive method when either the sample size or the noise level is increased. It is well known in measurement error literature that the performance of the estimation procedure heavily depends on the signal to noise ratio, or the ratio of $\sigma_x^2$ and $\sigma_u^2$. The signal to noise ratios in the previous simulation studies are $10$ and $4$. We also conducted some simulation studies with signal to noise ratio changed to $40$ and $16$ This resulted in improved performance of all three methods, with the naive method sometimes providing better results than the SIMEX and EX methods, which was not unexpected, since such high signal to noise ratios imply the effect of measurement error is nearly negligible.

  As mentioned in the beginning, we used the average of the estimates from $10$ independent data sets as the final estimate of the regression function. For illustration purposes, in Figure \ref{fig1} to Figure \ref{fig9}, for each simulation setup, we present the fitted the regression curves for $n=200$, $\sigma_u^2=0.25$, and $B=50$ for SIMEX, from all three methods, with the true regression function as the reference. For completeness, in each figure, the four small plots show the fitted regression curves from four different data sets, and the large plot shows the fitted regression curve based on the averages.

  \begin{figure}[h!]
   \centering
   \caption{Naive Estimate: $g(x)=x\sin(x)$}
   \label{fig1}
   \includegraphics[width=5.5in,height=2.7in]{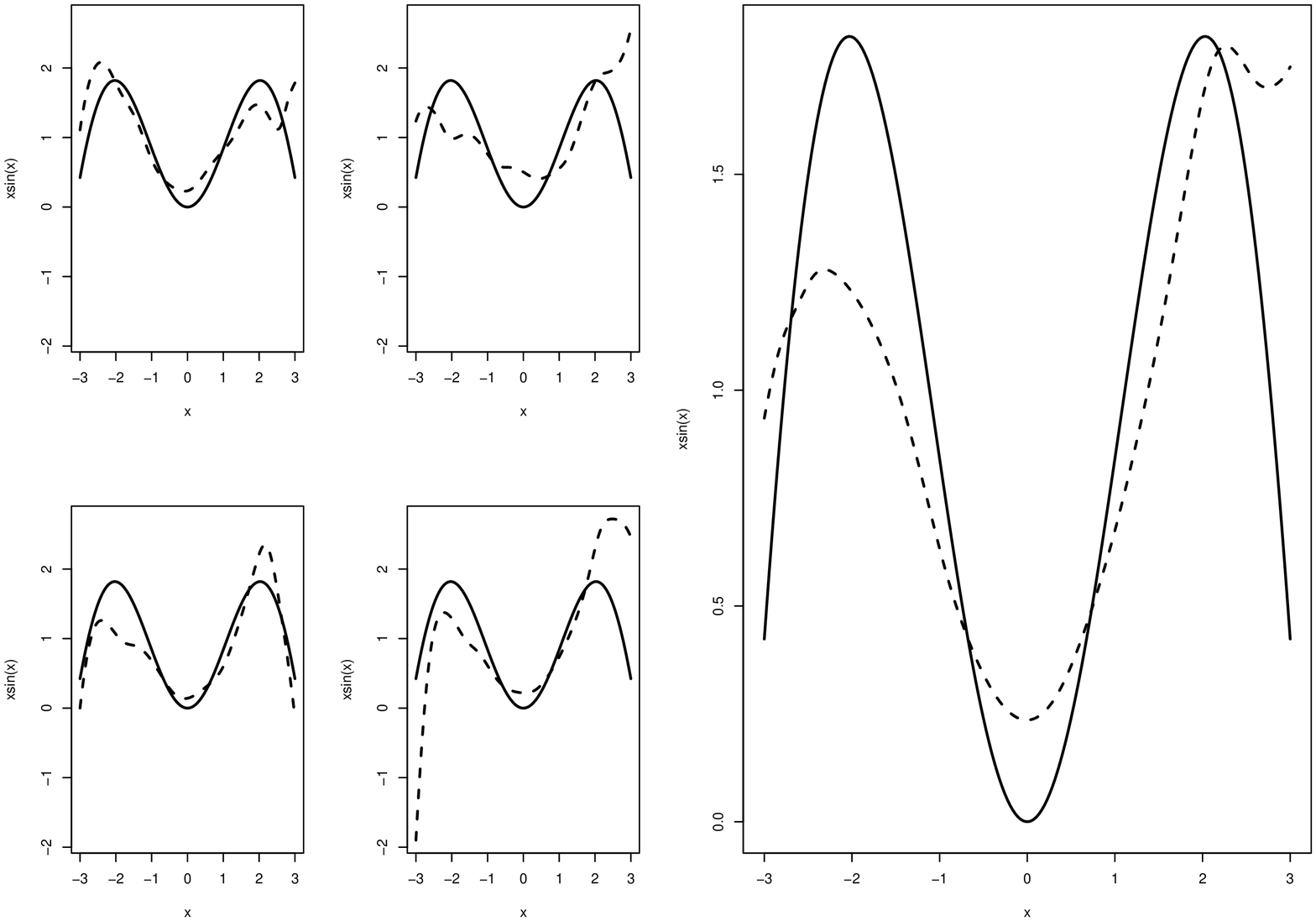}
  \end{figure}

  \begin{figure}[h!]
   \centering
   \caption{Naive Estimate: $g(x)=x^2$}
   \label{fig2}
   \includegraphics[width=5.5in,height=2.7in]{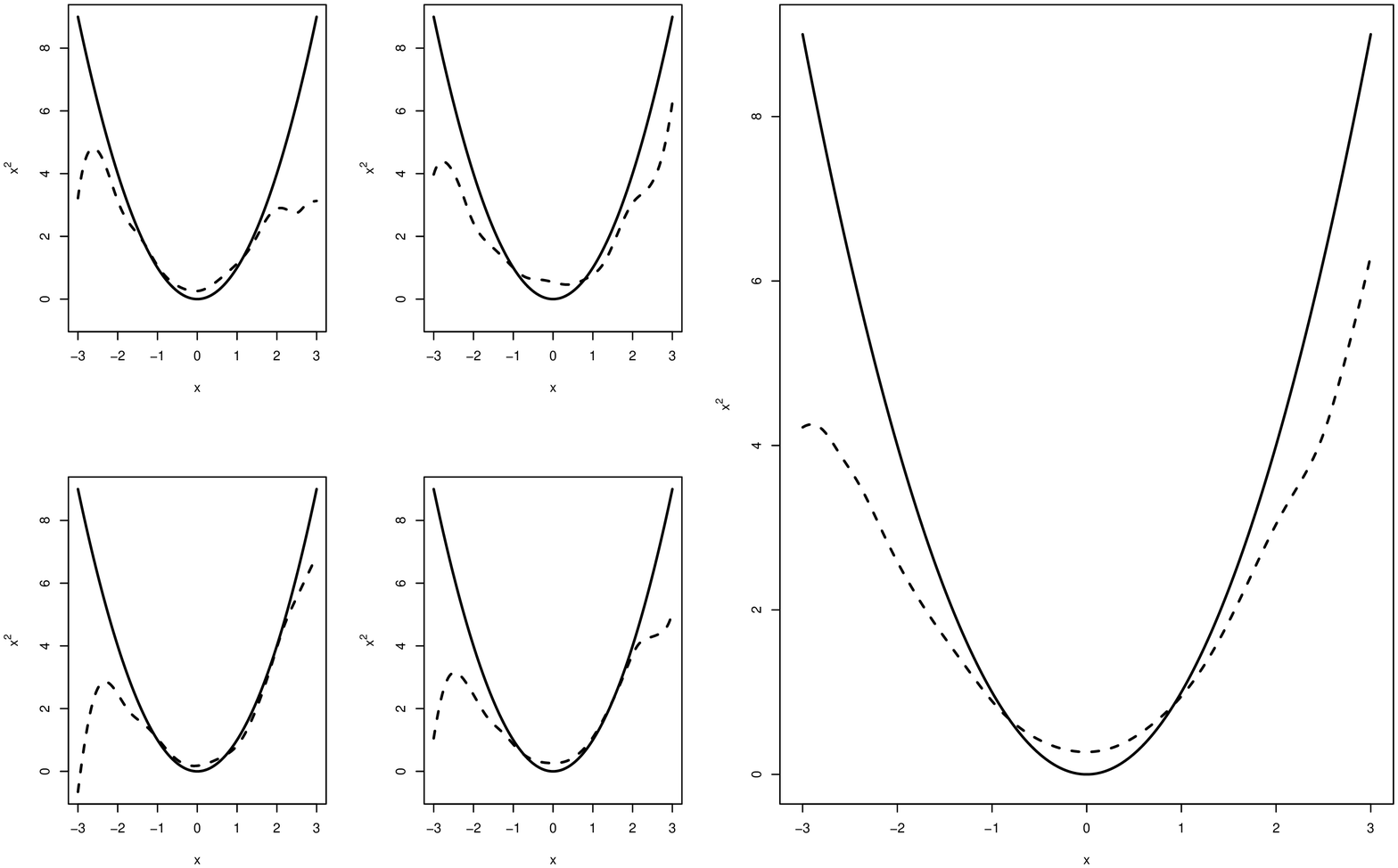}
  \end{figure}

  \begin{figure}[h!]
   \centering
   \caption{Naive Estimate: $g(x)=\exp(x)$}
     \label{fig3}
   \includegraphics[width=5.5in,height=2.7in]{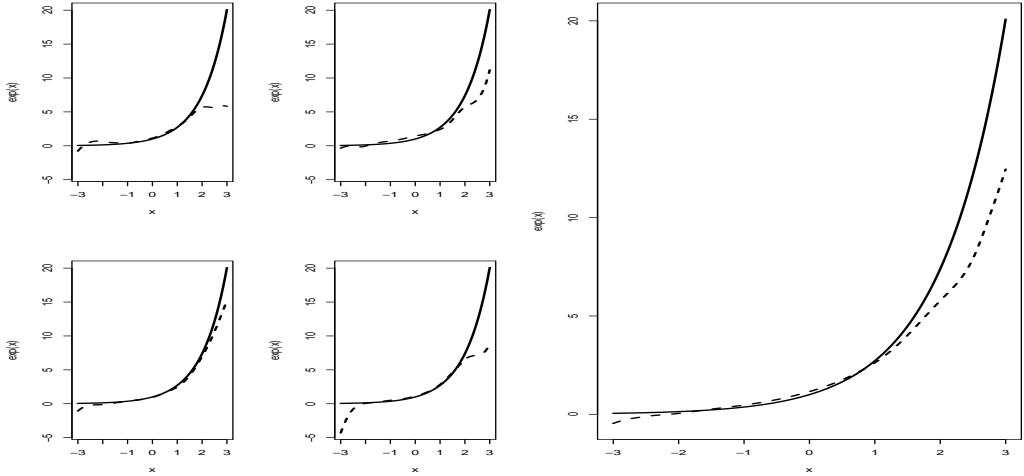}
  \end{figure}

  \begin{figure}[h!]
   \centering
   \caption{SIMEX Estimate: $g(x)=x\sin(x)$}  \label{fig4}
   \includegraphics[width=5.5in,height=2.7in]{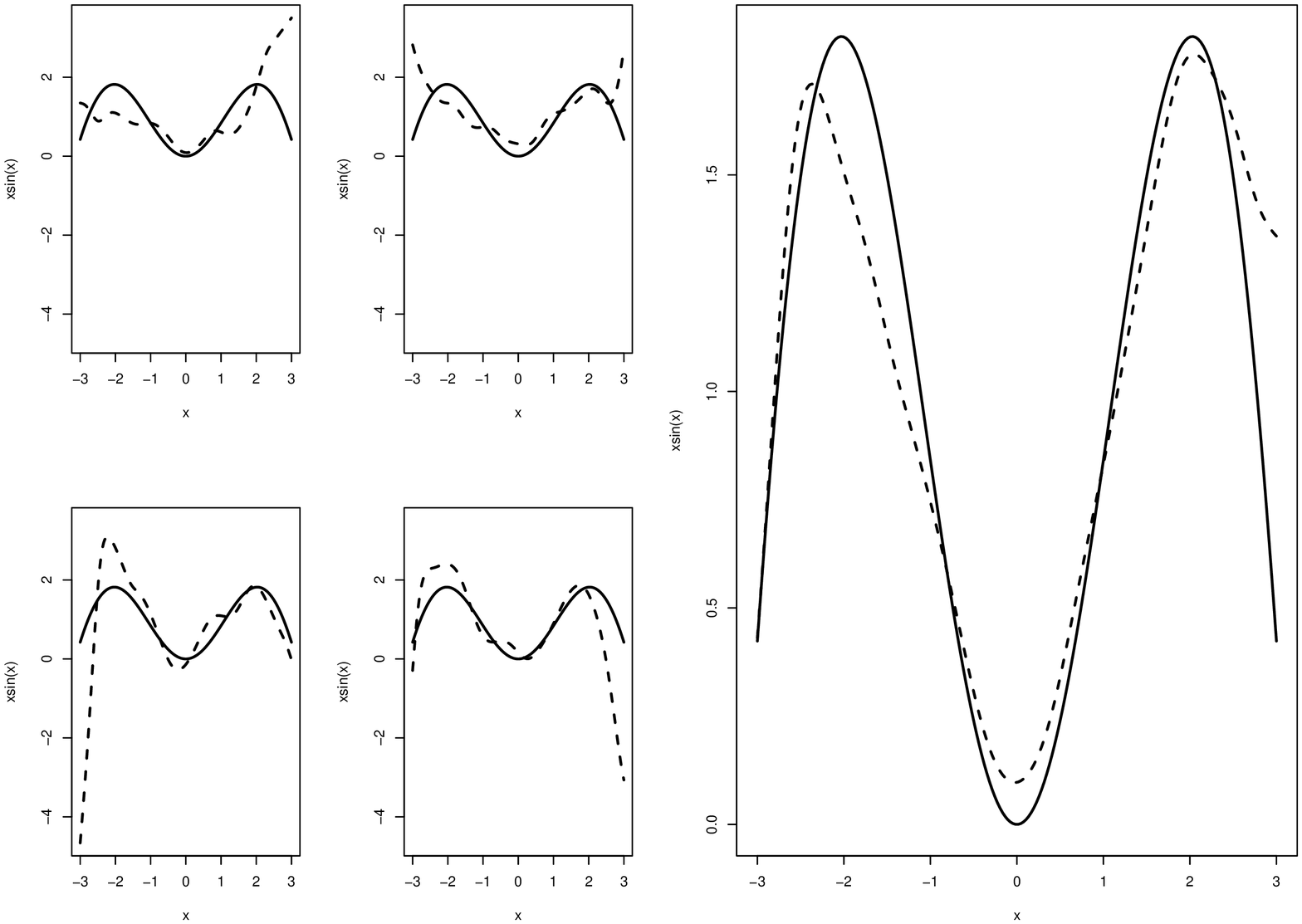}
  \end{figure}

  \begin{figure}[h!]
   \centering
   \caption{SIMEX Estimate: $g(x)=x^2$}
     \label{fig5}
   \includegraphics[width=5.5in,height=2.7in]{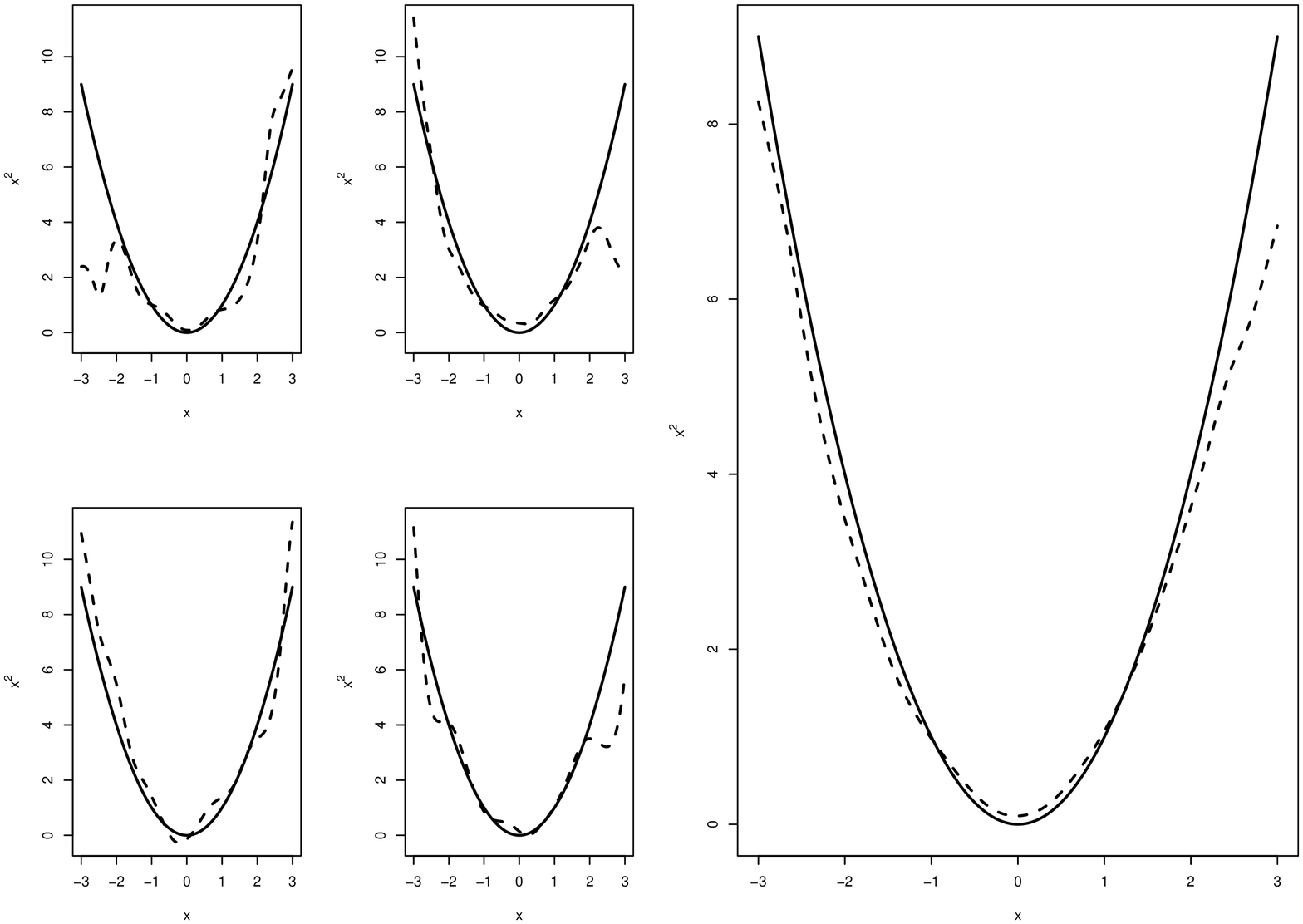}
  \end{figure}

  \begin{figure}[h!]
   \centering
   \caption{SIMEX Estimate: $g(x)=\exp(x)$}
  \label{fig6}
   \includegraphics[width=5.5in,height=2.7in]{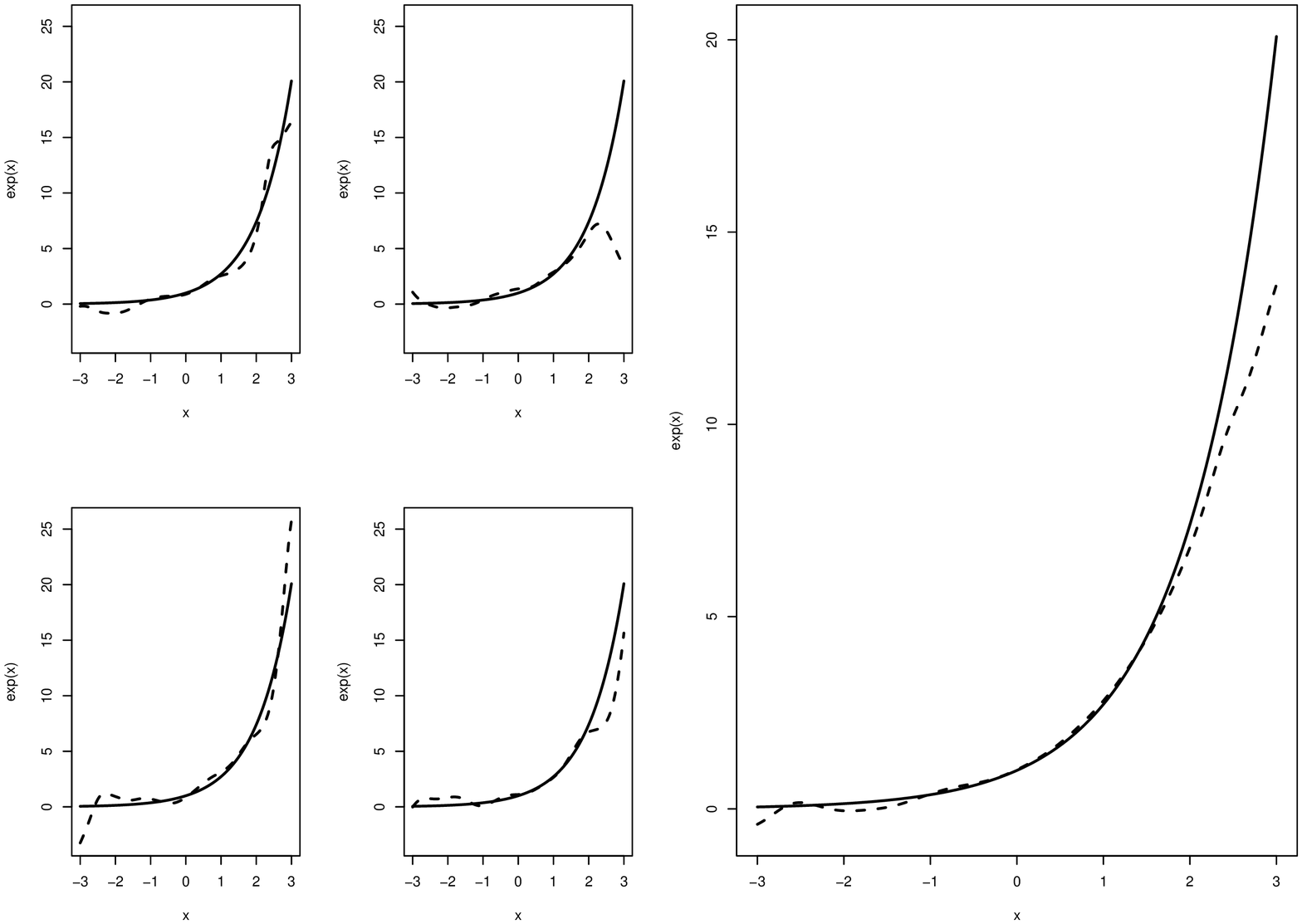}
  \end{figure}

  \begin{figure}[h!]
   \centering
   \caption{EX Estimate: $g(x)=x\sin(x)$}
  \label{fig7}
   \includegraphics[width=5.5in,height=2.7in]{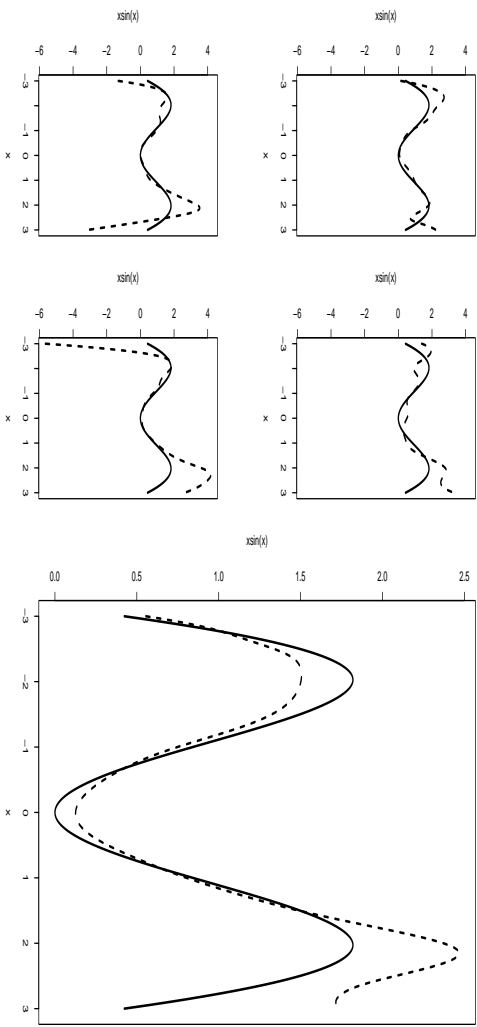}
  \end{figure}

  \begin{figure}[h!]
   \centering
   \caption{EX Estimate: $g(x)=x^2$}
     \label{fig8}
   \includegraphics[width=5.5in,height=2.7in]{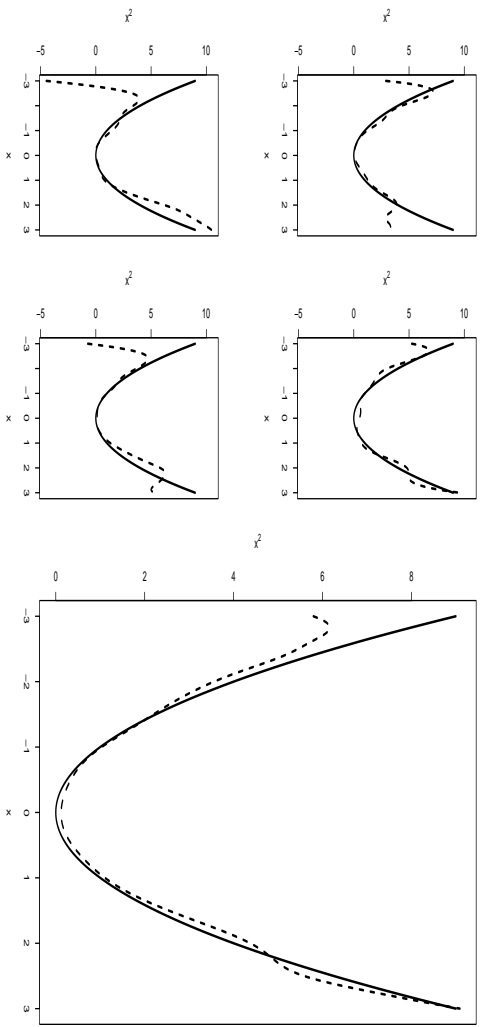}
  \end{figure}

  \begin{figure}[h!]
   \centering
   \caption{EX Estimate: $g(x)=\exp(x)$}
  \label{fig9}
   \includegraphics[width=5.5in,height=2.7in]{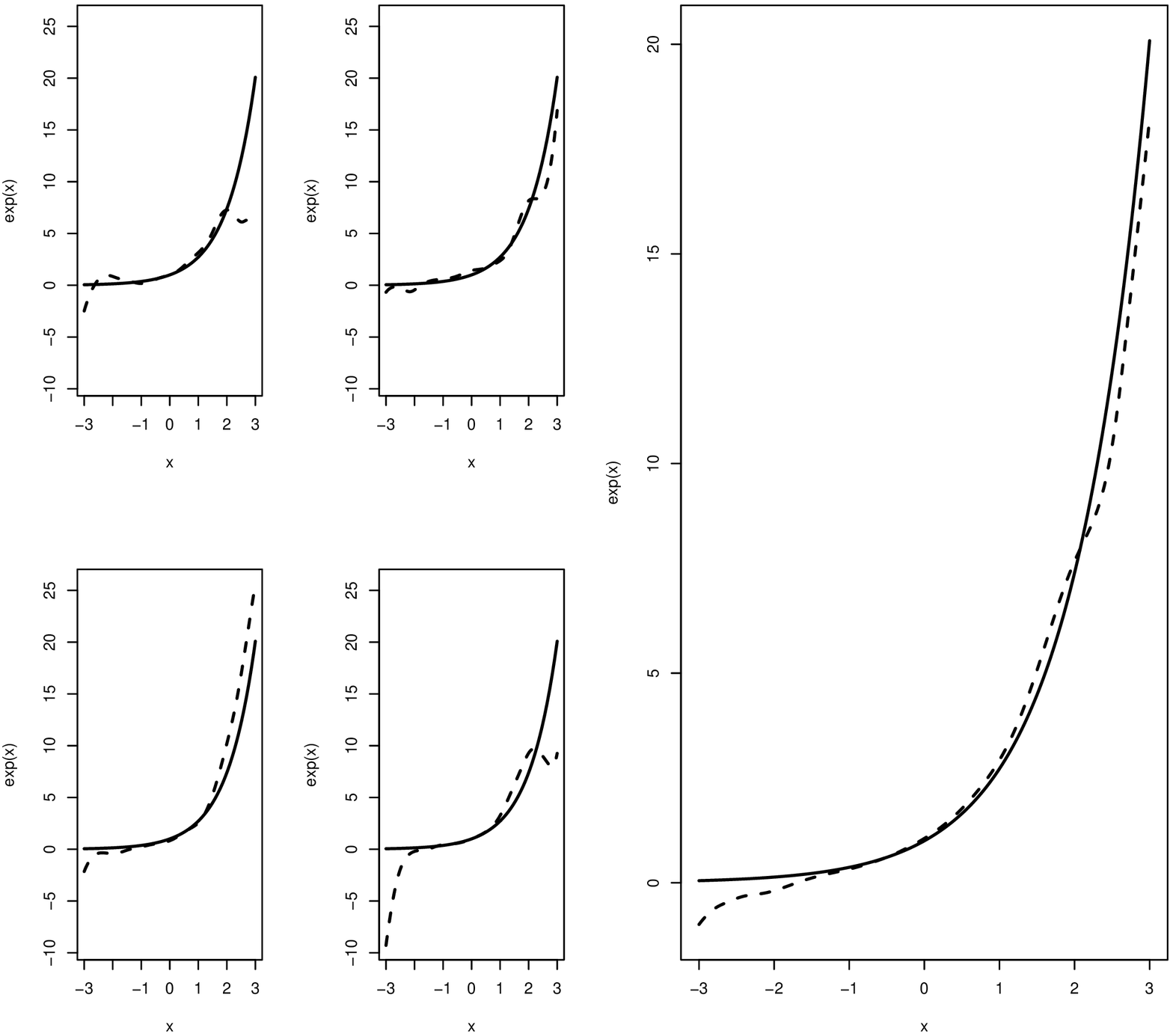}
  \end{figure}

\subsection{Real Data Application: NHANES Data Set}

   To determine the relationship between the serum 25-hydroxyvitamin D (25(OH)D) and the long term vitamin D average intake,  \cite{brenna2017} analyzed a data set from the National Health and Nutrition Examination Survey (NHANES), and used a nonlinear function for modeling the regression mean of 25(OH)D on the long term vitamin D average intake. In this section, we apply the proposed estimation procedure on a subset of the 2009-2010 NHANES study. The selected data set contains dietary records of $806$ Mexican-American females. The long term vitamin $D$ average intake ($X$) is not measured directly, instead, two independent daily observations of vitamin D intake are collected. Let $W_{ji}$ be the vitamin D intake from the $i$-th subject on the $j$-th time, and we assume that the additive structures $W_{ji}=X_i+U_{ji}$ hold for all $i=1,2,\ldots,806$, $j=1,2$. We use $W_i=(W_{1i}+W_{2i})/2$ to represent the observed vitamin intake, and by assuming that $U_{1i}$ and $U_{2i}$ are independently and identically distributed, we can estimate the standard deviation of the measurement error $U$ by the sample standard deviation of the differences $(W_{1i}-W_{2i})/2$, $i=1,2,\ldots,n$. As in \cite{brenna2017}, we also apply a square root transformation on the 25(OH)D which results in a more symmetric structure, but the Shapiro normal test reports a $p$-value of $0.04$, indicating that the transformed 25(OH)D values, denoted as $Y$, is still not normal.

   We adopt the local linear estimator to fit the regression function of $Y$ against $X$ to capture the mean regression function using the Naive, SIMEX ($B=200$) and the proposed EX methods. Three fitted regression functions with the bandwidth $h=n^{-1/5}$, together with the scatter plots of $Y$ against $W$, are plotted in Figure \ref{fig10}. In Figure \ref{fig10}, the solid line is the fitted EX regression function, the dashed line is the fitted regression function using the classical SIMEX, and the dotted line is the fitted regression curve using the naive method. Clearly the naive estimator captures the central structure of the raw data, as expected. The fitted regression function from the classical SIMEX nearly overlaps the proposed EX estimator. Compared with the naive regression, the SIMEX and the EX procedures provide relatively conservative fitted 25(OH)D values when the vitamin D intake values are small, which might be interpreted as an evidence of the subjects under-reporting their vitamin D intakes. Because fewer data points on the upper end, so we truncated the graph when the observed vitamin D intake is bigger than 15, therefore more caution should be paid when interpreting the trend on the right. More scientific explanations from the analysis need to consult with experts on nutrition studies. The computation times for each of the three methods are, $0.209$ seconds for Naive, $0.839$ seconds for the EX, and $209.58$ seconds for the classical SIMEX. Again, one can see that the proposed EX method is more efficient than the classical SIMEX.

  \begin{figure}[h!]
  \centering
  \caption{Naive, SIMEX and EX estimates}
  \label{fig10}
   \includegraphics[width=5.5in,height=4in]{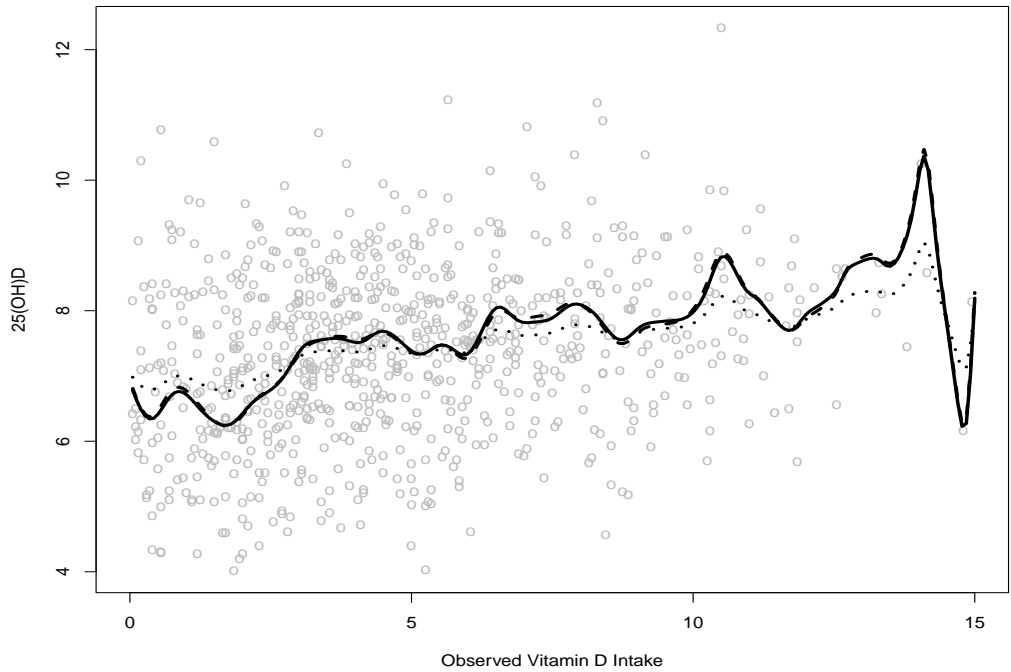}
  \end{figure}

 \section{Discussion}\label{sec7}

  Instead of taking the conditional expectation of the estimator based on the pseudo-data or following the three steps in the classical SIMEX algorithm, the proposed EX method applies the conditional expectation directly to the target function to be optimized based on the pseudo-data, thus successfully bypassing the simulation step. Both the simulation studies and the real data applications indicate the EX algorithm is more effective than the classical SIMEX, as evidenced by less computation time and smaller MSEs. In Section \ref{sec3}, we discussed the main difference between $\hat g_n(x;\lambda)$ and $\tilde g_n(x;\lambda)$, but a more detailed comparison can be made by a heuristic argument as shown below. Define $\bs=(s_0,s_1,s_2,t_0,t_1)^T$, and a function
    $$F(\bs)=F(s_0,s_1,s_2,t_0,t_1)=\frac{s_2t_0-s_1t_1}{s_2s_0-s_1^2},$$
 then a Taylor expansion at $\tilde\bs=(\tilde s_0, \tilde s_1, \tilde s_2, \tilde t_0, \tilde t_1)^T$ up to order 2 leads to
   $$
    F(\bs)\dot=F(\tilde\bs)+\frac{\partial F(\bs)}{\partial\bs}\Big|_{\bs=\tilde\bs}^T(\bs-\tilde\bs)+\frac{1}{2}
    (\bs-\tilde\bs)^T\frac{\partial^2 F(\bs)}{\partial\bs\partial\bs^T}\Big|_{\bs=\tilde\bs}(\bs-\tilde\bs).
   $$
 Let $s_j=S_{jn}(x,\bv)$, $\tilde s_j=\tilde S_{jn}(x)$ for $j=0,1,2$, and $t_j=T_{jn}(x,\bv)$, $\tilde t_j=\tilde T_{jn}(x)$ for $j=0,1$. Then it is easy to see that $\tilde g_n(x;\lambda)=E[F(\bs)|\bD]$, and $\hat g_n(x;\lambda)=F(\tilde\bs)$. Therefore, from the above Taylor expansion, we can see that $\tilde g_n(x;\lambda)-\hat g_n(x;\lambda)$ approximately equals
   \begin{eqarray*}
    E\left((\bs-\tilde\bs)^T\frac{\partial^2 F(\bs)}{\partial\bs\partial\bs^T}\Big|_{\bs=\tilde\bs}(\bs-\tilde\bs)\Bigg|\bD\right)
    =\mbox{trace}\left(\frac{\partial^2 F(\bs)}{\partial\bs\partial\bs^T}\Big|_{\bs=\tilde\bs}E\left((\bs-\tilde\bs)(\bs-\tilde\bs)^T|\bD\right)\right).
   \end{eqarray*}\\ \vskip -0.3in \noindent
 Note that the matrix $E\left((\bs-\tilde\bs)(\bs-\tilde\bs)^T|\bD\right)$ is nonnegative definite, so it is sufficient to consider the expectation of each entry only in the matrix to determine its order. We can show that all $25$ terms are of the order $O(1/nh)$ for $\lambda\geq 0$. This implies that, for $\lambda>0$, the estimators $\hat g_n(x;\lambda)$ and $\tilde g_n(x;\lambda)$ are equivalent, in the sense of having the same asymptotic distribution, if $nh^4\to 0$, and for $\lambda=0$, they are equivalent if $nh^5\to 0$. All the necessary computations supporting these claims can be found in the supplement materials.

 \subsection{On the extrapolation function}

  As we discussed in Section 5, theoretically it is impossible to specify the true form of the extrapolation function in the nonparametric regression setups. However, if the regression function $g$ has a parametric form, then according to the above discussion, we can indeed nail down the extrapolation function. To see this, consider the power function $x^p$ with some $p\geq 1$, with $X$ still assumed to be $N(0,\sigma_x^2)$. Then some algebra leads to
   $$
     \int t^p\phi\left(t,\frac{x\sigma_x^2}{(\lambda+1)\sigma_u^2+\sigma_x^2},
      \frac{(\lambda+1)\sigma_u^2\sigma_x^2}{(\lambda+1)\sigma_u^2+\sigma_x^2}
      \right)dt=\sum_{j=0}^{[p/2]}\binom{p}{2j}(2j-1)!!x^{p-2j}\frac{\sigma_x^{2p}\sigma_u^{4j}(\lambda+1)^{2j}}
      {[(\lambda+1)\sigma_u^2+\sigma_x^2]^{p}}.
   $$
 This implies that for a polynomial regression function $g$ of order $p$, if $X$ is normal, then the extrapolation function can be taken as a polynomial function of order $p$. Without loss of generality, assume that $H(\lambda)=s^T(\lambda)\balpha$, where $s(\lambda)=(1,\lambda,\lambda^2,\cdots,\lambda^p)^T$. Then $\balpha$ can be estimated by the minimizer of $L(\balpha)=\sum_{j=0}^K[\hat g_{\lambda_j}(x)-s^T(\lambda_j)\balpha]^2$. In fact, the minimizer 
   $
     \hat\balpha=\left[\sum_{j=0}^Ks(\lambda_j)s^T(\lambda_j)\right]^{-1}\sum_{j=0}^K\hat g_{\lambda_j}(x)s(\lambda_j).
   $
 Similar to \cite{car1999}, we have, for $nh^{5}\to 0$, $\sqrt{nh}(\hat\alpha-\alpha)$ is asymptotically normal with mean
   $
     \mu(x,h) =h^2B(x,0)\left[\sum_{j=0}^Ks(\lambda_j)s^T(\lambda_j)\right]^{-1}s(0),
   $
 and covariance matrix
   $$
     \tau^2(x)=\Delta_0(x)\left[\sum_{j=0}^Ks(\lambda_j)s^T(\lambda_j)\right]^{-1}s(0)s^T(0)\left[\sum_{j=0}^Ks(\lambda_j)s^T(\lambda_j)\right]^{-1}.
   $$
 Thus, for the SIMEX estimator $s^T(-1)\hat\balpha$, we have
 $\sqrt{nh}(s^T(-1)\hat\balpha-g(x)-s^T(-1)\mu(x,h))$ converges to $N\left(0, \tau^2(x)\right)$ in distribution.
 Of course, the discussion only has some theoretical significance. If we knew in advance that $g$ has a parametric form, we would not estimate it using the nonparametric methods.

 \section{Appendix}

   To prove Theorem \ref{thm0}, we need find out the expectations and variances of each component appearing in $\hat g_n(x;\lambda)$ defined in (\ref{eq3.8}). The calculation is facilitated by the following lemmas.

  \begin{lem}\label{lem1} Let $a, c$ be any positive constants. Suppose that for any $x$ in the support of $f_X$, $m'(t+x), m''(t+x)\in L_1(\phi(\cdot; 0,c))$ and are continuous as functions of $t$. Then, as $h\to 0$, we have
      $$
        \int \phi(t; x, ah^2+c)m(t)dt=\int\phi(t; x,c)m(t)dt+\frac{ah^2}{2}\int m''(t)\phi(t;x,c)dt+o(h^2).
      $$
  Furthermore, if
      $$
        \frac{\partial^j{\int m(t+x)\phi(t;0,c)dt}}{{\partial x^j}}=\int \frac{\partial^j{m(t+x)}}{{\partial x^j}}\phi(t;0,c)dt,\quad j=1,2,
      $$
  then we have
      $$
        \int \phi(t; x, ah^2+c)m(t)dt=\int\phi(t-x;0,c)m(t)dt+\frac{ah^2}{2}\cdot\frac{\partial^2{\int m(t)\phi(t-x;0,c)dt}}{{\partial x^2}}+o(h^2).
      $$
  \end{lem}\vskip 0.1in

 In Lemma \ref{lem1}, if we take $c=0$, then
   \begin{equation*}
      \int \phi(t; x, ah^2)m(t)dt=m(x)+\frac{ah^2}{2}m''(x)+o(h^2).
   \end{equation*}

  \begin{proof}[The proof of Lemma \ref{lem1}.]
   Note that%
   \begin{eqarray*}
   \frac{\partial m(x+u\sqrt{ah^2+c})}{\partial h}\Big|_{h=0}=0,\quad\quad \frac{\partial^2 m(x+u\sqrt{ah^2+c})}{\partial h^2}\Big|_{h=0}=m'(x+u\sqrt{c})\frac{au}{\sqrt{c}}.
  \end{eqarray*}\\ \vskip -0.3in \noindent
 Therefore, using Taylor expansion,
  \begin{eqarray}
    &&\int \phi(t; x, ah^2+c)m(t)dt=\int \phi(u;0,1)m(x+u\sqrt{ah^2+c})du\nonumber\\
    &=&\int \phi(u;0,1)\left[m(x+u\sqrt{c})+m'(x+u\sqrt{c})\frac{auh^2}{2\sqrt{c}}\right]du+o(h^2)\nonumber\\
    &=&\int \phi(t-x;0,c)\left[m(t)+m'(t)\frac{ah^2(t-x)}{2c}\right]dt+o(h^2).\label{eq8.1}
  \end{eqarray}
 Note that under the condition of $m'(t+x)\in L_1(\phi(\cdot; 0,c))$ for any $x\in\R$, we can get
  $$
   \frac{1}{c}\int\phi(t-x;0,c)m'(t)(t-x)dt=\int m''(t)\phi(t-x;0,c)dt.
  $$
 This, together with (\ref{eq8.1}), implies the first expansion.

 For the second expansion, notice that under the derivative-integration exchangeability condition, we have
   $$
    \int m''(t)\phi(t-x;0,c)dt=\int \frac{\partial^2m(t+x)}{\partial t^2}\phi(t;0,c)dt
    =\int \frac{\partial^2m(t+x)}{\partial x^2}\phi(t;0,c)dt.
   $$
 \end{proof}

 The following lemma lists some facts about normal density functions which are used often in the proofs of our main results. For the sake of brevity, the proofs of these facts are omitted since they can be found in standard statistics books.

 \begin{lem}\label{lem2} For normal density function $\phi(u; \mu,\sigma^2)$ with mean $\mu$ and variance $\sigma^2$, we have
   \begin{eqarray*}
     &&\phi^2(u;\mu,\sigma^2)=\frac{1}{2\sqrt{\pi\sigma^2}}\phi\left(u;\mu,\frac{\sigma^2}{2}\right),\quad
     \phi^3(u;\mu,\sigma^2)=\frac{1}{2\sqrt{3}\pi\sigma^2}\phi\left(u;\mu,\frac{\sigma^2}{3}\right),\\
     &&\phi(u; \mu_1,\sigma_1^2)\phi(u,\mu_2,\sigma_2^2)=\phi(\mu_1-\mu_2;0,\sigma_1^2+\sigma_2^2)
               \phi\left(u;\frac{\sigma_1^2\mu_2+\sigma_2^2\mu_1}{\sigma_1^2
                   +\sigma_2^2},\frac{\sigma_1^2\sigma_2^2}{\sigma_1^2+\sigma_2^2}\right),\\
     &&\int \phi(u;\mu_1,\sigma_1^2)\phi(u;\mu_2,\sigma_2^2)du=
                \phi(\mu_1-\mu_2;0,\sigma_1^2+\sigma_2^2),\\
     &&\int u\phi(u;\mu_1,\sigma_1^2)\phi(u;\mu_2,\sigma_2^2)du=
                \frac{\sigma_1^2\mu_2+\sigma_2^2\mu_1}{\sigma_1^2+\sigma_2^2}\phi(\mu_1-\mu_2;0,\sigma_1^2+\sigma_2^2),\\
    &&\int u^2\phi(u;\mu_1,\sigma_1^2)\phi(u;\mu_2,\sigma_2^2)du=\left[
                \frac{\sigma_1^2\sigma_2^2}{\sigma_1^2+\sigma_2^2}+
                \left(\frac{\sigma_1^2\mu_2+\sigma_2^2\mu_1}{\sigma_1^2+\sigma_2^2}\right)^2\right]\phi(\mu_1-\mu_2;0,\sigma_1^2+\sigma_2^2).
   \end{eqarray*}\\ \vskip -0.3in \noindent
 \end{lem}

 Then using the above two lemmas, for the components $\tilde S_{jn}(x)$ or $\tilde s_j$ for $j=0,1,2$, $\tilde T_{ln}(x)$ for $l=0,1$, in the definition of $\hat g_n(x)$ given in (\ref{eq3.8}), we can get the following series of results on the asymptotic expansions of their expectations and variances. For brevity, in the proof, we denote
 $\delta_{jh}^2=h^2+(\lambda+j)\sigma_u^2$ for $j=0,1,2$.

 \begin{lem}\label{lem3}
   For $\tilde S_{n0}(x)$, we have
    \begin{eqarray*}
       E(\tilde S_{n0}(x))&=&f_{0,\lambda}(x)+h^2f_{0,\lambda}''(x)/2+o(h^2),\quad \lambda\geq 0,\\[0.1in]
       \mbox{Var}(\tilde S_{n0}(x))&=&\left\{
            \begin{array}{l}
             \displaystyle\frac{f_{0,\lambda/2}(x)}{2n\sqrt{\pi\lambda\sigma_u^2}}
            -\frac{f_{0,\lambda}^2(x)}{n}+O\left(\frac{h^2}{n}\right),\quad \lambda>0,\\[0.15in]
           \displaystyle\frac{f_{0,0}(x)}{2nh\sqrt{\pi}}
            -\frac{f_{0,0}^2(x)}{n}+O\left(\frac{h}{n}\right), \quad \lambda=0.
           \end{array}
           \right.
    \end{eqarray*}
 \end{lem}

 \begin{proof}[Proof of Lemma \ref{lem3}.] 
 By the independence of $X$ and $U$, and applying Lemma \ref{lem1} with $a=1$, $c=(\lambda+1)\sigma_u^2$ and $m(t)=f_X(t)$, for $\tilde S_{0n}$, we have
  \begin{eqarray*}
    &&E\phi(x,Z,\delta_{0h}^2)=\iint \phi(x;t+u,\delta_{0h}^2)\phi(u;0,\sigma_u^2)f_X(t)dudt\nonumber\\
     &=&\int\phi(t-x;0,(\lambda+1)\sigma_u^2)f_X(t)dt
        +\frac{h^2}{2}\int f_X''(t)\phi(t-x;0,(\lambda+1)\sigma_u^2)dt+o(h^2)\nonumber\\
     &=&f_{0,\lambda}(x)+h^2f_{0,\lambda}''(x)/2+o(h^2).   \label{eq2.4}
   \end{eqarray*}\\ \vskip -0.3in \noindent
 Also, applying Lemma (\ref{lem1}) with $a=1/2$, $c=(\lambda+2)\sigma_u^2/2$, and $m(t)=f_X(t)$, we have
   \begin{eqarray}
     &&E\phi^2(x;Z,\delta_{0h}^2)=\iint \phi^2(x;t+u,\delta_{0h}^2)\phi(u;0,\sigma_u^2)f_X(t)dudt\nonumber\\
       &=&\frac{1}{2\sqrt{\pi\delta_{0h}^2}}\int \phi\left(t;x,\frac{\delta_{2h}^2}{2}\right)f_X(t)dt
     =\frac{f_{0,\lambda/2}(x)}{2\sqrt{\pi\delta_{0h}^2}}
     +\frac{h^2f_{0,\lambda/2}''(x)}{8\sqrt{\pi\delta_{0h}^2}}+o\left(\frac{h^2}{2\sqrt{\pi\delta_{0h}^2}}\right).
     \label{eq8.2}
   \end{eqarray}
Therefore, we have
   \begin{eqarray*}
     &&\mbox{var}\left[\tilde S_{0n}(x)\right]
       =\frac{1}{n}\left\{E\phi^2(x;Z,\delta_{0h}^2)-[E\phi(x;Z,\delta_{0h}^2)]^2\right\} \\
     &=&\frac{4f_{0,\lambda/2}(x)+h^2f_{0,\lambda/2}''(x)}{8n\sqrt{\pi\delta_{0h}^2}}
        -\frac{f_{0,\lambda}^2(x)+f_{0,\lambda}(x)f_{0,\lambda}''(x)h^2}{n}+o\left(\frac{h^2}{2n\sqrt{\pi\delta_{0h}^2}}\right).\nonumber
   \end{eqarray*}
 This concludes the proof of Lemma \ref{lem3}.
  \end{proof}

 \begin{lem}\label{lem4}
   For $\tilde S_{n1}(x)$, we have
     \begin{eqarray*}
           E(\tilde S_{n1}(x))&=&\frac{h^2}{(\lambda+1)\sigma_u^2}
              f_{1,\lambda}(x)-\frac{xh^2}{(\lambda+1)\sigma_u^2}f_{0,\lambda}(x)+o\left(h^2\right),\quad \lambda\geq 0,\\[0.1in]
           \mbox{Var}(\tilde S_{n1}(x))&=&\left\{
            \begin{array}{l}
             \displaystyle\frac{h^4}{2n(\lambda+2)^2\sigma_u^4\sqrt{\pi\lambda\sigma_u^2}}
           \left[f_{2,\lambda/2}(x)-xf_{1,\lambda/2}(x)+x^2f_{0,\lambda/2}(x)\right]+\\[0.2in]
            \displaystyle\quad \frac{h^4}{2n\lambda(\lambda+2)\sigma_u^2\sqrt{\pi\lambda\sigma_u^2}}f_{0,\lambda/2}(x)
           -\frac{h^4}{n(\lambda+1)^2\sigma_u^4}[f_{1,\lambda}(x)-xf_{0,\lambda}(x)]^2,\quad \lambda>0\\[0.2in]
           \displaystyle \frac{h}{4n\sqrt{\pi}} f_{0,0}(x)+o\left(\frac{h}{n}\right), \quad\lambda=0.
           \end{array}
         \right.
         \end{eqarray*}
 \end{lem}

 \begin{proof}[Proof of Lemma \ref{lem4}.]
 Applying Lemma \ref{lem1} with $a=1$, $c=(\lambda+1)\sigma_u^2$, $m(t)=f_X(t)$ and $tf_X(t)$, and Lemma \ref{lem2}, we have
   \begin{eqarray*}
    &&E(Z-x)\phi(x,Z,\delta_{0h}^2)=\iint (t+u-x)\phi(x;t+u,\delta_{0h}^2)\phi(u;0,\sigma_u^2)f_X(t)dudt\nonumber\\
       &=&\int (t-x)\left[\int \phi(u; x-t,\delta_{0h}^2)\phi(u,0,\sigma_u^2)du\right]f_X(t)dt\\
    &&   +\int\left[\int u\phi(u; x-t,\lambda\sigma_u^2)\phi(u; 0,\sigma_u^2)du\right]f_X(t)dt \nonumber\\
     &=&\int\phi(t-x;0,\delta_{1h}^2)(t-x)f_X(t)dt-\frac{\sigma_u^2}{\delta_{1h}^2}\int
         \phi(t-x;0,\delta_{1h}^2)(t-x)f_X(t)dt\nonumber\\
     &=&\frac{\delta_{0h}^2}{\delta_{1h}^2}\int\phi(t-x;0,\delta_{1h}^2)tf_X(t)dt-\frac{\delta_{0h}^2x}{\delta_{1h}^2}\int\phi(t-x;0,\delta_{1h}^2)f_X(t)dt\\
     &=&\frac{\delta_{0h}^2}{\delta_{1h}^2}
     \left[f_{1,\lambda}(x)+\frac{h^2}{2}f_{1,\lambda}''(x)\right]-\frac{x\delta_{0h}^2}{\delta_{1h}^2}
     \left[f_{0,\lambda}(x)+\frac{h^2}{2}f_{0,\lambda}''(x)\right]+o\left(\frac{\delta_{0h}^2h^2}{2\delta_{1h}^2}\right).
   \end{eqarray*}\\ \vskip -0.3in \noindent
 We also have
   \begin{eqarray*}
    &&E(Z-x)^2\phi^2(x,Z,\delta_{0h}^2)=\iint (t+u-x)^2\phi^2(x;t+u,\delta_{0h}^2)\phi(u;0,\sigma_u^2)f_X(t)dudt\nonumber\\
    &=&\frac{1}{2\sqrt{\pi\delta_{0h}^2}}\iint (t+u-x)^2\phi\left(u;x-t,\frac{\delta_{0h}^2}{2}\right)\phi(u;0,\sigma_u^2)f_X(t)dudt\nonumber\\
    &=&\frac{1}{2\sqrt{\pi\delta_{0h}^2}}\iint (t+u-x)^2\phi\left(t,x,\frac{\delta_{2h}^2}{2}\right)\phi\left(u;
    \frac{2\sigma_u^2(x-t)}{\delta_{2h}^2}, \frac{\sigma_u^2\delta_{0h}^2}{\delta_{2h}^2}\right)f_X(t)dudt\nonumber\\
    &=&\left(\frac{1}{2\sqrt{\pi\delta_{0h}^2}}-\frac{1}{\sqrt{\pi\delta_{0h}^2}}\cdot\frac{2\sigma_u^2}{\delta_{2h}^2}\right)
    \int (t-x)^2\phi\left(t,x,\frac{\delta_{2h}^2}{2}\right)f_X(t)dt\nonumber\\
    &&+\frac{1}{2\sqrt{\pi\delta_{0h}^2}}\int
    \left(\frac{\sigma_u^2\delta_{0h}^2}{\delta_{2h}^2}+
    \left[\frac{2\sigma_u^2(x-t)}{\delta_{2h}^2}\right]^2
    \right)
    \phi\left(t,x,\frac{\delta_{2h}^2}{2}\right)f_X(t)dt\nonumber\\
    &=&\frac{1}{2\sqrt{\pi\delta_{0h}^2}}\left(1-\frac{2\sigma_u^2}{\delta_{2h}^2}\right)^2
     \int (t-x)^2\phi\left(t,x,\frac{\delta_{2h}^2}{2}\right)f_X(t)dt+\frac{\sigma_u^2\delta_{0h}}{2\delta_{2h}^2\sqrt{\pi}}
    \int\phi\left(t,x,\frac{\delta_{2h}^2}{2}\right)f_X(t)dt\\
    &=&\frac{\delta_{0h}^3}{2\delta_{2h}^4\sqrt{\pi}}
       \int (t-x)^2\phi\left(t,x,\frac{\delta_{2h}^2}{2}\right)f_X(t)dt+
       \frac{\sigma_u^2\delta_{0h}^2}{2\delta_{2h}^2\sqrt{\pi\delta_{0h}^2}}\int
    \phi\left(t,x,\frac{\delta_{2h}^2}{2}\right)f_X(t)dt\\
    &=&\frac{\delta_{0h}^3}{2\delta_{2h}^4\sqrt{\pi}}\left(
    \left[f_{2,\lambda/2}(x)+\frac{h^2}{4}f_{2,\lambda/2}''(x)+o(h^2)\right]
     -x
    \left[f_{1,\lambda/2}(x)+\frac{h^2}{4}f_{1,\lambda/2}''(x)+o(h^2)\right]\right)\\
    &&+\left(\frac{x^2\delta_{0h}^3}{2\delta_{2h}^4\sqrt{\pi}}+\frac{\sigma_u^2\delta_{0h}}{2\delta_{2h}^2\sqrt{\pi}}\right)
      \left[f_{0,\lambda/2}(x)+\frac{h^2}{4}f_{0,\lambda/2}''(x)+o(h^2)\right].
   \end{eqarray*}\\ \vskip -0.3in \noindent
 Therefore, we have
  \begin{eqarray*}
    &&\mbox{var}\left[\tilde S_{1n}(x)\right]=\frac{h^4\left\{E(Z-x)^2\phi^2(x;Z,\delta_{0h}^2)-[E(Z-x)\phi(x;Z,\delta_{0h}^2)]^2\right\}}{n\delta_{0h}^4}=\\
    &&\hskip 0.1in \frac{h^4}{2n\delta_{2h}^4\sqrt{\pi\delta_{0h}^2}}\left(
    \left[f_{2,\lambda/2}(x)+\frac{h^2}{4}f_{2,\lambda/2}''(x)+o(h^2)\right]-x\left[f_{1,\lambda/2}(x)+\frac{h^2}{4}f_{1,\lambda/2}''(x)+o(h^2)\right]\right)\\
    &&\hskip 0.1in +\left(\frac{h^4x^2}{2\delta_{2h}^4 n\sqrt{\pi\delta_{0h}^2}}+\frac{\sigma_u^2 h^4}{2n\delta_{2h}^2\sqrt{\pi}\delta_{0h}^3}\right)
      \left[f_{0,\lambda/2}(x)+\frac{h^2}{4}f_{0,\lambda/2}''(x)+o(h^2)\right]\\
   &&\hskip 0.1in -\frac{1}{n}\Bigg(\frac{h^2}{\delta_{1h}^2}
     \left[f_{1,\lambda}(x)+\frac{h^2}{2}f_{1,\lambda}''(x)\right]-\frac{xh^2}{\delta_{1h}^2}
     \left[f_{0,\lambda}(x)+\frac{h^2}{2}f_{0,\lambda}''(x)\right]+o\left(\frac{h^2}{2\delta_{1h}^2}\right)\Bigg)^2.
  \end{eqarray*}\\ \vskip -0.3in \noindent
 This concludes the proof of Lemma \ref{lem4}.
 \end{proof}

 \begin{lem}\label{lem5}
   For $\tilde S_{n2}(x)$, we have
    \begin{eqarray*}
      E(\tilde S_{n2}(x))&=&h^2f_{0,\lambda}(x)+o(h^2),\quad \lambda\geq 0,\\[0.1in]
      \mbox{Var}(\tilde S_{n2}(x))&=&\left\{
       \begin{array}{l}
        \displaystyle\frac{h^4}{n}\left[\frac{f_{0,\lambda/2}(x)}{2\sqrt{\pi\lambda\sigma_u^2}}-f_{0,\lambda}^2(x)\right]+o\left(\frac{h^4}{n}\right),\quad
        \lambda>0,\\[0.2in]
       \displaystyle\frac{3h^3f_{0,0}(x)}{8n\sqrt{\pi}}+o\left(\frac{h^3}{n}\right), \quad \lambda=0.
       \end{array}
       \right.
    \end{eqarray*}
 \end{lem}

 \begin{proof}[Proof of Lemma \ref{lem5}.] 
 Note that
   \begin{eqarray*}
    &&E(Z-x)^2\phi(x,Z,\delta_{0h}^2)=\iint (t+u-x)^2\phi(x;t+u,\delta_{0h}^2)\phi(u;0,\sigma_u^2)f_X(t)dudt\nonumber\\
    &=&\iint (t+u-x)^2\phi\left(t,x,\delta_{1h}^2\right)\phi\left(u;
    \frac{\sigma_u^2(x-t)}{\delta_{1h}^2}, \frac{\sigma_u^2\delta_{0h}^2}{\delta_{1h}^2}\right)f_X(t)dudt\nonumber\\
    &=&\int (t-x)^2\phi\left(t,x,\delta_{1h}^2\right)f_X(t)dt-\frac{2\sigma_u^2}{\delta_{1h}^2}\int (t-x)^2\phi\left(t,x,\delta_{1h}^2\right)f_X(t)dt\nonumber\\
    &&+\int
    \left(\frac{\sigma_u^2\delta_{0h}^2}{\delta_{1h}^2}+
    \left[\frac{\sigma_u^2(x-t)}{\delta_{1h}^2}\right]^2
    \right)
    \phi\left(t,x,\delta_{1h}^2\right)f_X(t)dt\nonumber\\
    &=&\frac{\delta_{0h}^4}{\delta_{1h}^4}
       \int (t-x)^2\phi\left(t,x,\delta_{1h}^2\right)f_X(t)dt+\frac{\sigma_u^2\delta_{0h}^2}{\delta_{1h}^2}\int
    \phi\left(t,x,\delta_{1h}^2\right)f_X(t)dt\\
    &=&\frac{\delta_{0h}^4}{\delta_{1h}^4}
    \left[f_{2,\lambda}(x)+\frac{h^2}{2}f_{2,\lambda}''(x)+o(h^2)-2x
    \left(f_{1,\lambda}(x)+\frac{h^2}{2}f_{1,\lambda}''(x)+o(h^2)\right)\right.\\
    &&\left.+x^2   \left(f_{0,\lambda}(x)+\frac{h^2}{2}f_{0,\lambda}''(x)+o(h^2)\right)\right]+ \frac{\sigma_u^2\delta_{0h}^2}{\delta_{1h}^2}\left[f_{0,\lambda}(x)+\frac{h^2}{2}f_{0,\lambda}''(x)+o(h^2)\right].
   \end{eqarray*}\\ \vskip -0.3in \noindent
 Therefore, we get
  \begin{eqarray*}
    &&E[\tilde S_{2n}(x)]=\frac{h^4}{\delta_{0h}^4}\Bigg(\left[\frac{\delta_{0h}^2}{\delta_{1h}^2}\right]^2
    \left[f_{2,\lambda}(x)+\frac{h^2}{2}f_{2,\lambda}''(x)+o(h^2)
    -2x \left(f_{1,\lambda}(x)+\frac{h^2}{2}f_{1,\lambda}''(x)+o(h^2)\right)\right.\\
    &&\hskip 0.5in \left.+x^2
      \left(f_{0,\lambda}(x)+\frac{h^2}{2}f_{0,\lambda}''(x)+o(h^2)\right)\right]+ \frac{\sigma_u^2\delta_{0h}^2}{\delta_{1h}^2}\left[f_{0,\lambda}(x)+\frac{h^2}{2}f_{0,\lambda}''(x)+o(h^2)\right]\Bigg)\\
    &&\hskip 0.5in +\frac{\lambda\sigma_u^2h^2}{\delta_{0h}^2} \left[f_{0,\lambda}(x)+\frac{h^2}{2}f_{0,\lambda}''(x)+o(h^2)\right].
  \end{eqarray*}\\ \vskip -0.3in \noindent
 Then we have to calculate $E(Z-x)^4\phi^2(x,Z,\delta_{0h}^2)$. Recall that for $u\sim N(\mu,\sigma_u^2)$, we have
 $Eu^3=3\mu\sigma_u^2+\mu^3$, $Eu^4=3\sigma^4+6\mu^2\sigma_u^2+\mu^4$. So,
 \begin{eqarray*}
    &&E(Z-x)^4\phi^2(x,Z,\delta_{0h}^2)=\iint (t+u-x)^4\phi^2(u; x-t,\delta_{0h}^2)\phi(u;0,\sigma_u^2)f_X(t)dudt\nonumber\\
    &=&\frac{1}{2\sqrt{\pi\delta_{0h}^2}}\iint (t+u-x)^4\phi\left(u;x-t,\frac{\delta_{0h}^2}{2}\right)\phi(u;0,\sigma_u^2)f_X(t)dudt\nonumber\\
    &=&\frac{1}{2\sqrt{\pi\delta_{0h}^2}}\iint (t+u-x)^4\phi\left(t,x,\frac{\delta_{2h}^2}{2}\right)\phi\left(u;
    \frac{2\sigma_u^2(x-t)}{\delta_{2h}^2}, \frac{\sigma_u^2\delta_{0h}^2}{\delta_{2h}^2}\right)f_X(t)dudt\nonumber\\
    &=&\frac{1}{2\sqrt{\pi\delta_{0h}^2}}\int (t-x)^4\phi\left(t,x,\frac{\delta_{2h}^2}{2}\right)f_X(t)dt-
    \frac{4\sigma_u^2}{\delta_{2h}^2\sqrt{\pi\delta_{0h}^2}}\int (t-x)^4\phi\left(t,x,\frac{\delta_{2h}^2}{2}\right)f_X(t)dt\nonumber\\
    &&+\frac{6}{2\sqrt{\pi\delta_{0h}^2}}\int (t-x)^2\phi\left(t,x,\frac{\delta_{2h}^2}{2}\right)\left[
    \left(\frac{2\sigma_u^2(x-t)}{\delta_{2h}^2}\right)^2+ \frac{\sigma_u^2\delta_{0h}^2}{\delta_{2h}^2}\right]f_X(t)dt\nonumber\\
    &&+\frac{4}{2\sqrt{\pi\delta_{0h}^2}}\int (t-x)\phi\left(t,x,\frac{\delta_{2h}^2}{2}\right)\left[
    3\frac{\sigma_u^2\delta_{0h}^2}{\delta_{2h}^2}\frac{2\sigma_u^2(x-t)}{\delta_{2h}^2}+ \left[\frac{2\sigma_u^2(x-t)}{\delta_{2h}^2}\right]^3\right]f_X(t)dt\nonumber\\
    &&+\frac{1}{2\sqrt{\pi\delta_{0h}^2}}\int \phi\left(t,x,\frac{\delta_{2h}^2}{2}\right)\left[
      3\left(\frac{\sigma_u^2\delta_{0h}^2}{\delta_{2h}^2}+\left(\frac{2\sigma_u^2(x-t)}
      {\delta_{2h}^2}\right)^2\right)^2-
      2\left(\frac{2\sigma_u^2(x-t)}{\delta_{2h}^2}\right)^4\right]dt
   \end{eqarray*}\\ \vskip -0.3in \noindent
  It can be further written as
   \begin{eqarray*}
    &&\left(\frac{1}{2\sqrt{\pi\delta_{0h}^2}}-\frac{4\sigma_u^2}{\delta_{2h}^2\sqrt{\pi\delta_{0h}^2}}\right)\int (t-x)^4\phi\left(t,x,\frac{\delta_{2h}^2}{2}\right)f_X(t)dt\nonumber\\
    &&\hskip 0.3in +\frac{6}{2\sqrt{\pi\delta_{0h}^2}}\int (t-x)^2\phi\left(t,x,\frac{\delta_{2h}^2}{2}\right)\left[
    \left(\frac{2\sigma_u^2(x-t)}{\delta_{2h}^2}\right)^2+ \frac{\sigma_u^2\delta_{0h}^2}{\delta_{2h}^2}\right]f_X(t)dt\nonumber\\
    &&\hskip 0.3in +\frac{4}{2\sqrt{\pi\delta_{0h}^2}}\int (t-x)\phi\left(t,x,\frac{\delta_{2h}^2}{2}\right)\left[
    3\frac{\sigma_u^2\delta_{0h}^2}{\delta_{2h}^2}\frac{2\sigma_u^2(x-t)}{\delta_{2h}^2}+ \left[\frac{2\sigma_u^2(x-t)}{\delta_{2h}^2}\right]^3\right]f_X(t)dt\nonumber\\
    &&\hskip 0.3in+\frac{1}{2\sqrt{\pi\delta_{0h}^2}}\int \phi\left(t,x,\frac{\delta_{2h}^2}{2}\right)\left[
      3\left(\frac{\sigma_u^2\delta_{0h}^2}{\delta_{2h}^2}+\left(\frac{2\sigma_u^2(x-t)}{\delta_{2h}^2}\right)^2\right)^2
      -2
      \left(\frac{2\sigma_u^2(x-t)}{\delta_{2h}^2}\right)^4\right]dt
   \end{eqarray*}\\ \vskip -0.3in \noindent
 Let
  \begin{eqarray*}
    A(h;\lambda)&=&\frac{1}{2\sqrt{\pi\delta_{0h}^2}}
    \left[\frac{\delta_{0h}^2}{\delta_{2h}^2}\right]^4, \quad B(h;\lambda)=\frac{3\sigma_u^2\sqrt{\delta_{0h}^2}}{\delta_{2h}^2}\left[1-
     \frac{4\sigma_u^2}{\delta_{2h}^2}+\frac{4\sigma_u^4}{\delta_{2h}^4}\right]
                    =\frac{3\sigma_u^2\delta_{0h}^5}{\sqrt{\pi}\delta_{2h}^6},\\
    C(h;\lambda)&=&\frac{3\sigma_u^4\delta_{0h}^2\sqrt{\delta_{0h}^2}}{2\sqrt{\pi}\delta_{2h}^4}.
  \end{eqarray*}\\ \vskip -0.3in \noindent
 Then,
  \begin{eqarray*}
    &&E(Z-x)^4\phi^2(x,Z,\delta_{0h}^2)=A(h;\lambda)\int (t-x)^4\phi\left(t,x,\frac{\delta_{2h}^2}{2}\right)f_X(t)dt\\
    &&\hskip 0.3in +B(h;\lambda)\int (t-x)^2\phi\left(t,x,\frac{\delta_{2h}^2}{2}\right)f_X(t)dt+C(h;\lambda)\int \phi\left(t,x,\frac{\delta_{2h}^2}{2}\right)f_X(t)dt\\
    &=&A(h;\lambda)\int t^4\phi\left(t,x,\frac{\delta_{2h}^2}{2}\right)f_X(t)dt-4xA(h;\lambda)\int t^3\phi\left(t,x,\frac{\delta_{2h}^2}{2}\right)f_X(t)dt\\
    &&+(6x^2A(h;\lambda)+B(h;\lambda))\int t^2\phi\left(t,x,\frac{\delta_{2h}^2}{2}\right)f_X(t)dt\\
    &&-(4x^3A(h;\lambda)+2xB(h;\lambda))\int t\phi\left(t,x,\frac{\delta_{2h}^2}{2}\right)f_X(t)dt\\
    &&+[x^4A(h;\lambda)+x^2B(h;\lambda)+C(h;\lambda)]\int \phi\left(t,x,\frac{\delta_{2h}^2}{2}\right)f_X(t)dt\\
    &=&A(h;\lambda)\left[f_{4,\lambda/2}(x)+\frac{h^2}{4}f''_{4,\lambda/2}(x)+o(h^2)\right]-4xA(h;\lambda)\left[f_{3,\lambda/2}(x)+\frac{h^2}{4}f''_{3,\lambda/2}(x)+o(h^2)\right]\\
    &&+(6x^2A(h;\lambda)+B(h;\lambda))\left[f_{2,\lambda/2}(x)+\frac{h^2}{4}f''_{2,\lambda/2}(x)+o(h^2)\right]\\
    &&-(4x^3A(h;\lambda)+2xB(h;\lambda))\left[f_{1,\lambda/2}(x)+\frac{h^2}{4}f''_{1,\lambda/2}(x)+o(h^2)\right]\\
    &&+[x^4A(h;\lambda)+x^2B(h;\lambda)+C(h;\lambda)]\left[f_{0,\lambda/2}(x)+\frac{h^2}{4}f''_{0,\lambda/2}(x)+o(h^2)\right].
  \end{eqarray*}\\ \vskip -0.3in \noindent
 Summarizing above derivations eventually leads to
  \begin{eqarray*}
    \mbox{var}[\tilde S_{2n}(x)]&=&\frac{h^8}{n\delta_{0h}^8}\left[E(Z-x)^4\phi^2(x; Z,\delta_{0h}^2)
    -(E(Z-x)^2\phi(x; Z,\delta_{0h}^2))^2\right]\\
    && +\frac{\lambda^2\sigma_u^4h^4}{n\delta_{0h}^4}\left[E\phi^2(x; Z,\delta_{0h}^2)
    -(E\phi(x; Z,\delta_{0h}^2))^2\right]+\frac{2\lambda\sigma_u^2h^6}{n\delta_{0h}^6}\cdot\\
    &&\hskip 0.2in\left[E(Z-x)^2\phi^2(x; Z,\delta_{0h}^2)
    -E(Z-x)^2\phi(x; Z,\delta_{0h}^2)E\phi(x; Z,\delta_{0h}^2)\right].
  \end{eqarray*}\\ \vskip -0.3in \noindent
 This concludes the proof of Lemma \ref{lem5}.
 \end{proof}

 \begin{lem}\label{lem6}
   For $\tilde T_{n0}(x)$, we have
     \begin{eqarray*}
        E(\tilde T_{n0}(x))&=&g_{0,\lambda}(x)+\frac{h^2}{2}g''_{0,\lambda}(x)+o(h^2), \quad \lambda>0,\\
        \mbox{Var}(\tilde T_{n0}(x))&=& \left\{
        \begin{array}{l}
        \displaystyle\frac{1}{n}\left[\frac{1}{2\sqrt{\lambda\pi\sigma_u^2}}[G_{0,\lambda/2}(x)+H_{0,\lambda/2}(x)]-g_{0,\lambda}^2(x)\right]
                   +O\left(\frac{h^2}{n}\right),\quad\lambda>0,\\
        \displaystyle\frac{1}{2nh\sqrt{\pi}}[G_{0,0}(x)+H_{0,0}(x)]+o\left(\frac{1}{nh}\right),\quad\lambda=0.
         \end{array}
         \right.
     \end{eqarray*}
 \end{lem}

 \begin{proof}[Proof of Lemma \ref{lem6}.]  Note that
  \begin{eqarray*}
    &&E[Y\phi(x;Z,\delta_{0h}^2)]= E[(g(X)+\vep)\phi(x; Z,\delta_{0h}^2)]\\
    &=&E\left(E[(g(X)+\vep)\phi(x; Z,\delta_{0h}^2)\Big|X,U]\right)= E[g(X)\phi(x; Z,\delta_{0h}^2)] \\
    &=&\int g(t)f_X(t)\phi(x-t; 0, (\lambda+1)\sigma^2)dt
       +\frac{h^2}{2}\frac{\partial^2}{\partial x^2}\int g(t)f_X(t)\phi(x-t; 0, (\lambda+1)\sigma_u^2)dt+o(h^2).
  \end{eqarray*}\\ \vskip -0.3in \noindent
 Note that $\tau^2(X)=E(\vep^2|X)$, we also have
  \begin{eqarray*}
    &&E[Y^2\phi^2(x;Z,\delta_{0h}^2)]=E\left(E[(g(X)+\vep)^2\phi^2(x; Z,\delta_{0h}^2)\Big|X,U]\right)\\
    &=& E[(g^2(X)+\tau^2(X))\phi^2(x; Z,\delta_{0h}^2)]=\int [g^2(t)+\tau^2(t)]\phi^2(u; x-t, \delta_{0h}^2)\phi(u;0,\sigma_u^2)f_X(t)dt\\
    &=& \frac{1}{2\sqrt{\pi\delta_{0h}^2}}\int [g^2(t)+\tau^2(t)]\phi(x-t; 0, (\lambda+2)\sigma_u^2/2)f_X(t)dt\\
    && +\frac{h^2}{8\sqrt{\pi\delta_{0h}^2}}\frac{\partial^2}{\partial x^2}\int [g^2(t)+\tau^2(t)]\phi(x-t; 0, (\lambda+2)\sigma_u^2/2)f_X(t)dt+o\left(\frac{h^2}{2\sqrt{\pi\delta_{0h}^2}}\right).
  \end{eqarray*}\\ \vskip -0.3in \noindent
Therefore,
     \begin{eqarray*}
     \mbox{var}\left[\tilde T_{0n}(x)\right]
       &=&\frac{1}{n}\left\{E[Y^2\phi^2(x;Z,\delta_{0h}^2)]-(E[Y\phi(x;Z,\delta_{0h}^2)])^2\right\} \\
     &=&\frac{1}{2n\sqrt{\pi\delta_{0h}^2}}\int [g^2(t)+\tau^2(t)]\phi(x-t; 0, (\lambda+2)\sigma_u^2/2)f_X(t)dt\\
    && +\frac{h^2}{8n\sqrt{\pi\delta_{0h}^2}}\frac{\partial^2}{\partial x^2}\int [g^2(t)+\tau^2(t)]\phi(x-t; 0, (\lambda+2)\sigma_u^2/2)f_X(t)dt\\
    &&+o\left(\frac{h^2}{2n\sqrt{\pi\delta_{0h}^2}}\right) -\frac{1}{n}\bigg{[}\int g(t)f_X(t)\phi(x-t; 0, (\lambda+1)\sigma_u^2)dt\\
     &&
       +\frac{h^2}{2}\frac{\partial^2}{\partial x^2}\int g(t)f_X(t)\phi(x-t; 0, (\lambda+1)\sigma_u^2)dt+o(h^2)\bigg{]}^2.
   \end{eqarray*}\\ \vskip -0.3in \noindent
 This implies the result in Lemma \ref{lem6}.
 \end{proof}

 \begin{lem}\label{lem7}
   For $\tilde T_{n1}(x)$, we have
    \begin{eqarray*}
       E(\tilde T_{n1}(x))&=&\frac{h^2}{(\lambda+1)\sigma_u^2}
         g_{1,\lambda}(x)-\frac{xh^2}{(\lambda+1)\sigma_u^2}g_{0,\lambda}(x)+o\left(h^2\right),\quad \lambda\geq 0, \\[0.1in]
       \mbox{Var}(\tilde T_{n1}(x))&=&\\[0.2in]
         &&\hskip -0.7in\left\{
         \begin{array}{l}
         \displaystyle\frac{h^4}{n(\lambda+2)^2\sqrt{\lambda\pi}\sigma_u^5}\left[\frac{1}{2}[G_{2,\lambda/2}(x)+H_{2,\lambda/2}(x)]-x[G_{1,\lambda/2}(x)+H_{1,\lambda/2}(x)]\right]+\\[0.2in]
          \displaystyle\hskip 0.4in \frac{h^4}{n(\lambda+2)\sqrt{\lambda\pi}\sigma_u^3}\left[\frac{x^2}{(\lambda+2)\sigma_u^2}+\frac{1}{\lambda}\right]
            [G_{0,\lambda/2}(x)+H_{0,\lambda/2}(x)]-\\[0.2in]
           \displaystyle\hskip 0.4in
           \frac{h^4}{n(\lambda+1)^2\sigma_u^4}[g_{1,\lambda}(x)-xg_{0,\lambda}(x)]^2,\quad \lambda>0,\\[0.2in]
         \displaystyle\frac{h}{4n\sqrt{\pi}}\left[G_{0,0}(x)+H_{0,0}(x)\right]+o\left(\frac{h}{n}\right),\quad \lambda=0.
         \end{array}
         \right.
     \end{eqarray*}
 \end{lem}

  \begin{proof}[Proof of Lemma \ref{lem7}.]  Note that
\begin{eqarray*}
&&E\left[Y(Z-x)\phi(x;Z,\delta_{0h}^2)\right] =E\left[(g(X)+\vep)(Z-x)\phi(x; Z,\delta_{0h}^2)\right]\\
&=& E\left(E[(g(X)+\vep)(Z-x)\phi(x; Z,\delta_{0h}^2)\Big|X,U]\right)=E\left[g(X)(Z-x)\phi(x; Z,\delta_{0h}^2)\right] \\
&=& \iint g(t)(t+u-x)\phi(x;t+u,\delta_{0h}^2)\phi(u;0,\sigma_u^2)f_X(t)dudt\\
&=&\int g(t)(t-x)\left[\int \phi(u; x-t,\delta_{0h}^2)\phi(u,0,\sigma_u^2)du\right]f_X(t)dt\\
&&   +\int g(t)\left[\int u\phi(u; x-t,\lambda\sigma_u^2)\phi(u; 0,\sigma_u^2)du\right]f_X(t)dt   \\
&=&\int\phi(t,x,\delta_{1h}^2)g(t)(t-x)f_X(t)dt-\frac{\sigma_u^2}{\delta_{1h}^2}\int \phi(t,x,\delta_{1h}^2)g(t)(t-x)f_X(t)dt\\
&=&\frac{\delta_{0h}^2}{\delta_{1h}^2}\int\phi(t,x,\delta_{1h}^2)tg(t)f_X(t)dt-\frac{\delta_{0h}^2x}{\delta_{1h}^2}
 \int\phi(t,x,\delta_{1h}^2)g(t)f_X(t)dt\\
&=&\frac{\delta_{0h}^2}{\delta_{1h}^2}\Bigg[\int tg(t)f_X(t)\phi(t; x, (\lambda+1)\sigma_u^2)dt+\frac{h^2}{2}\frac{\partial^2}{\partial x^2}\int tg(t)f_X(t)\phi(t; x, (\lambda+1)\sigma_u^2)dt+o(h^2)\Bigg]\\
      && -\frac{\delta_{0h}^2x}{\delta_{1h}^2}\Bigg{[}\int g(t)f_X(t)\phi(t; x, (\lambda+1)\sigma_u^2)dt+\frac{h^2}{2}\frac{\partial^2}{\partial x^2}\int g(t)f_X(t)\phi(t; x, (\lambda+1)\sigma_u^2)dt+o(h^2)\Bigg]
\end{eqarray*}\\ \vskip -0.3in \noindent
Next, we see that
  \begin{eqarray*}
    &&E[Y^2(Z-x)^2\phi^2(x;Z,\delta_{0h}^2)]= E\left(E[(g(X)+\vep)^2(Z-x)^2\phi^2(x; Z,\delta_{0h}^2)\Big|X,U]\right)\\
    &=& E[(g^2(X)+\tau^2(X))(Z-x)^2\phi^2(x; Z,\delta_{0h}^2)]\\
    &=&\int [g^2(t)+\tau^2(t)]\int(t+u-x)^2\phi^2(x;t+u,\delta_{0h}^2)\phi(u;0,\sigma_u^2)f_X(t)dudt\nonumber\\
    &=& \frac{\delta_{0h}^3}{2\delta_{2h}^4\sqrt{\pi}}\bigg[\int [g^2(t)+\tau^2(t)]t^2f_X(t)\phi(x-t; 0, (\lambda+2)\sigma_u^2/2)dt + o(h^2)\bigg]\\
    && +\frac{h^2\delta_{0h}^3}{8\delta_{2h}^4\sqrt{\pi}}\frac{\partial^2}{\partial x^2}\int [g^2(t)+\tau^2(t)]t^2f_X(t)\phi(x-t; 0, (\lambda+2)\sigma_u^2/2)dt\\
    &&- \frac{x\delta_{0h}^3}{\delta_{2h}^4\sqrt{\pi}}\bigg[\int [g^2(t)+\tau^2(t)]tf_X(t)\phi(x-t; 0, (\lambda+2)\sigma_u^2/2)dt +o(h^2)\bigg]\\
    && -\frac{xh^2\delta_{0h}^3}{4\delta_{2h}^4\sqrt{\pi}}\frac{\partial^2}{\partial x^2}\int [g^2(t)+\tau^2(t)]tf_X(t)\phi(x-t; 0, (\lambda+2)\sigma_u^2/2)dt\\
    && + \frac{x^2\delta_{0h}^3}{2\delta_{2h}^4\sqrt{\pi}}\bigg[\int [g^2(t)+\tau^2(t)]f_X(t)\phi(x-t; 0, (\lambda+2)\sigma_u^2/2)dt + o(h^2)\bigg]\\
    && +\frac{x^2h^2\delta_{0h}^3}{8\delta_{2h}^4\sqrt{\pi}}\frac{\partial^2}{\partial x^2}\int [g^2(t)+\tau^2(t)]f_X(t)\phi(x-t; 0, (\lambda+2)\sigma_u^2/2)dt\\
&& + \frac{\sigma_u^2\delta_{0h}}{2\delta_{2h}^2\sqrt{\pi}}\int [g^2(t)+\tau^2(t)]f_X(t)\phi(x-t; 0, (\lambda+2)\sigma_u^2/2)dt\\
    && +\frac{\sigma_u^2h^2\delta_{0h}}{8\delta_{2h}^2\sqrt{\pi}}\frac{\partial^2}{\partial x^2}\int [g^2(t)+\tau^2(t)]f_X(t)\phi(x-t; 0, (\lambda+2)\sigma_u^2/2)dt+o\left(\frac{h^2\delta_{0h}}{2\delta_{2h}^2\sqrt{\pi}}\right)
  \end{eqarray*}\\ \vskip -0.3in \noindent
 Therefore,
  \begin{eqarray*}
     &&\mbox{var}\left[\tilde T_{1n}(x)\right]=
       \frac{h^4}{n\delta_{0h}^4}\left\{E[Y^2(Z-x)^2\phi^2(x;Z,\delta_{0h}^2)]-\bigg(E[Y(Z-x)\phi(x;Z,\delta_{0h}^2)]\bigg)^2\right\} \\
        &=& \frac{h^4}{2n\delta_{2h}^4\sqrt{\pi\delta_{0h}^2}}\bigg[\int [g^2(t)+\tau^2(t)]t^2f_X(t)\phi(x-t; 0, (\lambda+2)\sigma_u^2/2)dt +o(h^2)\bigg]\\
    &&\hskip -0.1in +\frac{h^6}{8n\delta_{2h}^4\sqrt{\pi\delta_{0h}^2}}\frac{\partial^2}{\partial x^2}\int [g^2(t)+\tau^2(t)]t^2f_X(t)\phi(x-t; 0, (\lambda+2)\sigma_u^2/2)dt\\
    &&\hskip -0.1in - \frac{xh^4}{n\delta_{2h}^4\sqrt{\pi\delta_{0h}^2}}\bigg[\int [g^2(t)+\tau^2(t)]tf_X(t)\phi(x-t; 0, (\lambda+2)\sigma_u^2/2)dt + o(h^2)\bigg]\\
    &&\hskip -0.1in  -\frac{xh^6}{4n\delta_{2h}^4\sqrt{\pi\delta_{0h}^2}}\frac{\partial^2}{\partial x^2}\int [g^2(t)+\tau^2(t)]tf_X(t)\phi(x-t; 0, (\lambda+2)\sigma_u^2/2)dt\\
    &&\hskip -0.1in  + \frac{x^2h^4}{2n\delta_{2h}^4\sqrt{\pi\delta_{0h}^2}}\bigg[\int [g^2(t)+\tau^2(t)]f_X(t)\phi(x-t; 0, (\lambda+2)\sigma_u^2/2)dt + o(h^2) \bigg]\\
    &&\hskip -0.1in  +\frac{x^2h^6}{8n\delta_{2h}^4\sqrt{\pi\delta_{0h}^2}}\frac{\partial^2}{\partial x^2}\int [g^2(t)+\tau^2(t)]f_X(t)\phi(x-t; 0, (\lambda+2)\sigma_u^2/2)dt\\
&&\hskip -0.1in  + \frac{\sigma_u^2h^4}{2n\delta_{2h}^2\delta_{0h}^3\sqrt{\pi}}\bigg[\int [g^2(t)+\tau^2(t)]f_X(t)\phi(x-t; 0, (\lambda+2)\sigma_u^2/2)dt + o(h^2)\bigg]\\
    &&\hskip -0.1in  +\frac{\sigma_u^2h^6}{8n\delta_{2h}^2\delta_{0h}^3\sqrt{\pi}}\frac{\partial^2}{\partial x^2}\int [g^2(t)+\tau^2(t)]f_X(t)\phi(x-t; 0, (\lambda+2)\sigma_u^2/2)dt\\
    &&\hskip -0.1in  -\frac{h^4}{n\delta_{1h}^4}\Bigg[\int tg(t)f_X(t)\phi(x-t; 0, (\lambda+1)\sigma_u^2)dt +\frac{h^2}{2}\frac{\partial^2}{\partial x^2}\int tg(t)f_X(t)\phi(x-t; 0, (\lambda+1)\sigma_u^2)dt\\
 &&\hskip -0.1in - x\int g(t)f_X(t)\phi(t; x, (\lambda+1)\sigma_u^2)dt -\frac{xh^2}{2}\frac{\partial^2}{\partial x^2}\int g(t)f_X(t)\phi(t; x, (\lambda+1)\sigma_u^2)dt+o(h^2)\Bigg]^2.
  \end{eqarray*}\\ \vskip -0.3in \noindent
  This concludes the proof of Lemma \ref{lem7}.
 \end{proof}

 \begin{proof}[Proof of Theorem \ref{thm1}.] 
  To verify the Lyapunov condition, we have to find out the asymptotic expansions of $Ev^2(x)$, and an upper bound for $E|v^3(x)|$. Note that
  \begin{eqarray*}
    Ev^2(x)&=& c_{0\lambda}^2E\xi_{0\lambda}^2(x)+c_{1\lambda}^2E\xi_{1\lambda}^2(x)+c_{2\lambda}^2E\xi_{2\lambda}^2(x)
              +d_{0\lambda}^2E\eta_{0\lambda}^2(x)+d_{1\lambda}^2E\eta_{1\lambda}^2(x)\\
           && +  2c_{0\lambda}c_{1\lambda}E[\xi_{0\lambda}(x)\xi_{1\lambda}(x)]
              +  2c_{0\lambda}c_{2\lambda}E[\xi_{0\lambda}(x)\xi_{2\lambda}(x)]
              +  2c_{0\lambda}d_{0\lambda}E[\xi_{0\lambda}(x)\eta_{0\lambda}(x)]\\
           && +  2c_{0\lambda}d_{1\lambda}E[\xi_{0\lambda}(x)\eta_{1\lambda}(x)]
              +  2c_{1\lambda}c_{2\lambda}E[\xi_{1\lambda}(x)\xi_{2\lambda}(x)]
              +  2c_{1\lambda}d_{0\lambda}E[\xi_{1\lambda}(x)\eta_{0\lambda}(x)]\\
           && +  2c_{1\lambda}d_{1\lambda}E[\xi_{1\lambda}(x)\eta_{1\lambda}(x)]
              +  2c_{2\lambda}d_{0\lambda}E[\xi_{2\lambda}(x)\eta_{0\lambda}(x)]
              +  2c_{2\lambda}d_{1\lambda}E[\xi_{2\lambda}(x)\eta_{1\lambda}(x)]\\
           && +  2d_{0\lambda}d_{1\lambda}E[\eta_{0\lambda}(x)\eta_{1\lambda}(x)].
  \end{eqarray*}\\ \vskip -0.3in \noindent
 Routine and tedious calculations show that when $\lambda>0$, except for
  \begin{eqarray*}
     E[\xi_{0\lambda}(x)\eta_{0\lambda}(x)]&=& \frac{1}{2\sqrt{\pi\lambda\sigma_u^2}}g_{0,\lambda/2}(x)-g_{0,\lambda}(x)f_{0,\lambda}(x)+O(h^2),
  \end{eqarray*}\\ \vskip -0.3in \noindent
 all other expectations of the cross products are of the order $O(h^2)$, which, together with the previous derivations with respect to $E\xi_{0\lambda}^2(x)$, $E\xi_{1\lambda}^2(x)$, $E\xi_{2\lambda}^2(x)$, $E\eta_{0\lambda}^2(x)$, $E\eta_{1\lambda}^2(x)$, we can obtain
  \begin{eqarray*}
    Ev^2(x)&=&c_{0\lambda}^2\left[\frac{1}{2\sqrt{\pi\lambda\sigma_u^2}}f_{0,\lambda/2}(x)-f_{0,\lambda}^2(x)\right]
       +d_{0\lambda}^2\left[\frac{1}{2\sqrt{\pi\lambda\sigma_u^2}}\{G_{0,\lambda/2}(x)+H_{0,\lambda/2}(x)\}-g_{0,\lambda}^2\right]\\
           && +2c_{0\lambda}d_{0\lambda}\left[\frac{1}{2\sqrt{\pi\lambda\sigma_u^2}}g_{0,\lambda/2}(x)-g_{0,\lambda}(x)f_{0,\lambda}(x)\right]
           +O(h^2).
  \end{eqarray*}\\ \vskip -0.3in \noindent
 When $\lambda=0$, except for $E[\xi_{00}(x)\eta_{00}(x)]=\frac{1}{2h\sqrt{\pi}}g_{0,0}(x)+O(h)$,
 all other expectations of the cross products are of the order $O(h)$, which, together with the previous derivations with respect to $E\xi_{00}^2(x)$, $E\xi_{10}^2(x)$, $E\xi_{20}^2(x)$, $E\eta_{00}^2(x)$, $E\eta_{10}^2(x)$, leads to
  \begin{eqarray*}
    Ev^2(x)&=&\frac{1}{2h\sqrt{\pi}}\left[c_{00}^2f_{00}(x)+d_{00}^2\{G_{00}(x)+H_{00}(x)\}+2c_{00}d_{00}g_{00}(x)\right]+O(h)\\
           &=&\frac{1}{2h\sqrt{\pi}}\left[\frac{G_{00}(x)+H_{00}(x)}{f_{00}^2(x)}-\frac{g_{00}^2(x)}{f_{00}^3(x)}\right]+O(h).
  \end{eqarray*}\\ \vskip -0.3in \noindent
 To find a proper order for $E|v(x)|^3$, we have to find the orders for the expectations
    \begin{eqarray*}
     && E(Z-x)\phi^2(x,Z,\delta_{0h}^2),\quad EY\phi^2(x,Z,\delta_{0h}^2),\quad EY(Z-x)\phi^2(x,Z,\delta_{0h}^2),\\
     && EY(Z-x)^2\phi^2(x,Z,\delta_{0h}^2),\quad EY^2(Z-x)\phi^2(x,Z,\delta_{0h}^2),\\
     && E\phi^3(x,Z,\delta_{0h}^2),\quad E|Z-x|^3\phi^3(x,Z,\delta_{0h}^2), \quad
     E|Z-x|^6\phi^3(x,Z,\delta_{0h}^2),\\
     && E|Y|^3\phi^3(x,Z,\delta_{0h}^2),\quad E|Y|^3|Z-x|^3\phi^3(x,Z,\delta_{0h}^2).
   \end{eqarray*}\\ \vskip -0.3in \noindent
 More complicated calculations show that
   \begin{eqarray*}
    E\phi^3(x,Z,\delta_{0h}^2)=\frac{1}{2\pi\delta_{0h}^2\sqrt{3}}
         \left[f_{0,\lambda/3}(x)+\frac{h^2}{6}f_{0,\lambda/3}''(x)+o(h^2)\right],
   \end{eqarray*}\\ \vskip -0.3in \noindent
 which is $O(1)$ when $\lambda>0$ and $O(1/h^2)$ when $\lambda=0$,
 and
   \begin{eqarray*}
   &&E|Z-x|^3\phi^3(x,Z,\delta_{0h}^2)\\
   &\leq&
   \frac{2}{\pi\sqrt{3}\delta_{0h}^2}\cdot \frac{8\sqrt{2}}{\sqrt{\pi}} \left(\frac{\sigma^2_u\delta_{0h}^2}{h^2+(\lambda+3)\sigma_u^2}\right)^\frac{3}{2}\int \phi\left(t,x,\frac{h^2+(\lambda+3)\sigma_u^2}{3}\right)f_X(t)dt\\
    && + \frac{2}{\pi\sqrt{3}\delta_{0h}^2} \cdot 4\left(\frac{3\sigma^2_u}{h^2+(\lambda+3)\sigma_u^2}\right)^3\int |x-t|^3\phi\left(t,x,\frac{h^2+(\lambda+3)\sigma_u^2}{3}\right)f_X(t)dt\\
 && + \frac{2}{\pi\sqrt{3}\delta_{0h}^2}\int |x-t|^3\phi\left(t,x,\frac{h^2+(\lambda+3)\sigma_u^2}{3}\right)f_X(t)dt
   \end{eqarray*}\\ \vskip -0.3in \noindent
 which is $O(1)$ when $\lambda>0$ and $O(1/h^2)$ when $\lambda=0$. We also have
 \begin{eqarray*}
     &&E|Z-x|^6\phi^3(x;Z,\delta_{0h}^2)\leq\frac{7680\sigma_u^6}{\pi\sqrt{3}\delta_{0h}^2} \left(\frac{\sigma^2_u\delta_{0h}^2}{h^2+(\lambda+3)\sigma_u^2}\right)^3\int \phi\left(t,x,\frac{h^2+(\lambda+3)\sigma_u^2}{3}\right)f_X(t)dt\\
     &&\hskip 0.3in + \frac{512}{\pi\sqrt{3}\delta_{0h}^2} \left(\frac{3\sigma^2_u}{h^2+(\lambda+3)\sigma_u^2}\right)^6\int |x-t|^6\phi\left(t,x,\frac{h^2+(\lambda+3)\sigma_u^2}{3}\right)f_X(t)dt\\
     &&\hskip 0.3in + \frac{16}{\pi\sqrt{3}\delta_{0h}^2}\int |x-t|^6\phi\left(t,x,\frac{h^2+(\lambda+3)\sigma_u^2}{3}\right)f_X(t)dt
	\end{eqarray*}\\ \vskip -0.3in \noindent
 which is $O(1)$ when $\lambda>0$ and $O(1/h^2)$ when $\lambda=0$. Denote $\delta(X)=E\left(|\epsilon|^3\Big|X\right)$,
then,
\begin{eqarray*}
 E|Y|^3\phi^3(x;Z,\delta_{0h}^2)\leq 4E|g(X)|^3 \phi^3(x;Z,\delta_{0h}^2) + 4E\delta(X)\phi^3(x;Z,\delta_{0h}^2).
\end{eqarray*}\\ \vskip -0.3in \noindent
 Eventually, we can show that
\begin{eqarray*}
     &&E|Y|^3\phi^3(x;Z,\delta_{0h}^2)\leq \frac{2}{\pi\sqrt{3}\delta_{0h}^2}\int \left[|g(t)|^3 + \delta(t)\right] \phi\left(t;x,\frac{h^2+(\lambda+3)\sigma_u^2}{3}\right)f_X(t)dt\nonumber
\end{eqarray*}\\ \vskip -0.3in \noindent
which is $O(1)$ for $\lambda>0$ and $O(1/h^2)$ for $\lambda=0$. Finally, for
For $E|Y|^3|Z-x|^3\phi^3(x;Z,\delta_{0h}^2)$, we can show that
\begin{eqarray*}
     &&E|Y|^3|Z-x|^3\phi^3(x;Z,\delta_{0h}^2)\nonumber\\
     &\leq&\frac{64\sqrt{2}}{\pi\sqrt{3\pi}\delta_{0h}^2} \left(\frac{\sigma^2_u\delta_{0h}^2}{h^2+(\lambda+3)\sigma_u^2}\right)^\frac{3}{2}\int \left[|g(t)|^3 + \delta(t)\right] \phi\left(t,x,\frac{h^2+(\lambda+3)\sigma_u^2}{3}\right)f_X(t)dt\\
     && +\left[ \frac{32}{\pi\sqrt{3}\delta_{0h}^2} \left(\frac{3\sigma^2_u}{h^2+(\lambda+3)\sigma_u^2}\right)^3+ \frac{8}{\pi\sqrt{3}\delta_{0h}^2}\right]\cdot\\
     && \int \left[|g(t)|^3 + \delta(t)\right]  |x-t|^3\phi\left(t,x,\frac{h^2+(\lambda+3)\sigma_u^2}{3}\right)f_X(t)dt,
	\end{eqarray*}\\ \vskip -0.3in \noindent
which is $O(1)$ when $\lambda>0$ and $O(1/h^2)$ when $\lambda=0$.

Therefore, when $\lambda>0$,
  $$
    \frac{\sum_{i=1}^n E|v_{i\lambda}(x)|^3}{\left(\sum_{i=1}^n Ev_{i\lambda}^2(x)\right)^{3/2}}=\frac{O(n)}{O(n^{3/2})}\to 0
  $$
as $n\to\infty$, and when $\lambda=0$,
  $$
    \frac{\sum_{i=1}^n E|v_{i\lambda}(x)|^3}{\left(\sum_{i=1}^n Ev_{i\lambda}^2(x)\right)^{3/2}}=\frac{O(n/h^2)}{O((n/h)^{3/2})}=O\left(\frac{1}{\sqrt{nh}}\right)\to 0.
  $$
So, by Lyapunov central limit theorem, we proved Theorem \ref{thm1}.
\end{proof}

 \begin{proof}[Proof of (\ref{eq3.3})-(\ref{eq3.5}).] By the normality assumption of $V$ and its independence from other random variables in the model, and the kernel function $K$ being the standard normal density, from Lemma \ref{lem2}, we have
   \begin{eqarray*}
     E[K_h(Z(\lambda)-x)|Y,Z]=\int \phi(v; x-Z,h^2)\phi(v; 0,\lambda\sigma_u^2)dv=\phi(x; Z, \delta_{0h}^2),
   \end{eqarray*}\\ \vskip -0.3in \noindent
 which is (\ref{eq3.3}). (\ref{eq3.4}) can be derived from the following algebra,
   \begin{eqarray*}
     && E[(Z(\lambda)-x)K_h(Z(\lambda)-x)|Y,Z]=\int (Z+v-x)\phi(v; x-Z,h^2)\phi(v; 0,\lambda\sigma_u^2)dv\\
     &=&(Z-x)\phi(x-Z; 0, \delta_{0h}^2)+\frac{\lambda\sigma_u^2(x-Z)}{\delta_{0h}^2}\phi(x-Z; 0, \delta_{0h}^2).
   \end{eqarray*}\\ \vskip -0.3in \noindent
 Finally, note that
   \begin{eqarray*}
     && E[(Z(\lambda)-x)^2K_h(Z(\lambda)-x)|Y,Z]=\int (Z+v-x)^2\phi(v; x-Z,h^2)\phi(v; 0,\lambda\sigma_u^2)dv\\
     &=&(Z-x)^2\int \phi(v; x-Z,h^2)\phi(v; 0,\lambda\sigma_u^2)dv
        +2(Z-x)\int v\phi(v; x-Z,h^2)\phi(v; 0,\lambda\sigma_u^2)dv\\
     && +\int v^2\phi(v; x-Z,h^2)\phi(v; 0,\lambda\sigma_u^2)dv\\
     &=&(Z-x)^2\phi(x-Z; 0, \delta_{0h}^2)-\frac{2\lambda\sigma_u^2(Z-x)^2}{\delta_{0h}^2}\phi(x-Z; 0, \delta_{0h}^2)\\
     && + \left[\frac{\lambda\sigma_u^2h^2}{\lambda\sigma_u^2+h^2}+\left(\frac{\lambda\sigma_u^2(x-Z)}{\lambda\sigma_u^2+h^2}\right)^2\right]
     \phi(x-Z; 0, \delta_{0h}^2),
   \end{eqarray*}\\ \vskip -0.3in \noindent
 this is exactly (\ref{eq3.5}).
 \end{proof}\vskip 0.2in

 \noindent{\it Acknowledgement:} Jianhong Shi's research is supported by the National Natural Science Foundation of China (Grant No. 12071267).

\bibliographystyle{elsarticle-harv}
\bibliography{simex}

\begin{thebibliography}{19}
\expandafter\ifx\csname natexlab\endcsname\relax\def\natexlab#1{#1}\fi
\providecommand{\url}[1]{\texttt{#1}}
\providecommand{\href}[2]{#2}
\providecommand{\path}[1]{#1}
\providecommand{\DOIprefix}{doi:}
\providecommand{\ArXivprefix}{arXiv:}
\providecommand{\URLprefix}{URL: }
\providecommand{\Pubmedprefix}{pmid:}
\providecommand{\doi}[1]{\href{http://dx.doi.org/#1}{\path{#1}}}
\providecommand{\Pubmed}[1]{\href{pmid:#1}{\path{#1}}}
\providecommand{\bibinfo}[2]{#2}
\ifx\xfnm\relax \def\xfnm[#1]{\unskip,\space#1}\fi
\bibitem[{Carroll et~al.(2006)Carroll, Rupper, Crainiceanu and
  Stefanski}]{carr2006}
\bibinfo{author}{Carroll, R.}, \bibinfo{author}{Rupper, D.},
  \bibinfo{author}{Crainiceanu, C.}, \bibinfo{author}{Stefanski, L.},
  \bibinfo{year}{2006}.
\newblock \bibinfo{title}{Measurement error in nonlinear models: a modern
  perspective}.
\newblock \bibinfo{edition}{2} ed., \bibinfo{publisher}{Chapman and Hall/CRC}.
\bibitem[{Carroll et~al.(1996)Carroll, Kuchenhoff, Lombard and
  Stefanski}]{carroll1996}
\bibinfo{author}{Carroll, R.J.}, \bibinfo{author}{Kuchenhoff, H.},
  \bibinfo{author}{Lombard, F.}, \bibinfo{author}{Stefanski, L.A.},
  \bibinfo{year}{1996}.
\newblock \bibinfo{title}{Asymptotics for the simex estimator in nonlinear
  measurement error models}.
\newblock \bibinfo{journal}{Journal of the American Statistical Association}
  \bibinfo{volume}{91}, \bibinfo{pages}{9}.
\bibitem[{Carroll et~al.(1999)Carroll, Maca and Ruppert}]{car1999}
\bibinfo{author}{Carroll, R.J.}, \bibinfo{author}{Maca, J.D.},
  \bibinfo{author}{Ruppert, D.}, \bibinfo{year}{1999}.
\newblock \bibinfo{title}{Nonparametric regression in the presence of
  measurement error}.
\newblock \bibinfo{journal}{Biometrika} \bibinfo{volume}{86},
  \bibinfo{pages}{541--554}.
\bibitem[{Carroll and Wang(2008)}]{carr2008}
\bibinfo{author}{Carroll, R.J.}, \bibinfo{author}{Wang, Y.},
  \bibinfo{year}{2008}.
\newblock \bibinfo{title}{Nonparametric variance estimation in the analysis of
  microarray data: a measurement error approach}.
\newblock \bibinfo{journal}{Biometrika} \bibinfo{volume}{95},
  \bibinfo{pages}{437--449}.
\bibitem[{Cook and Stefanski(1994)}]{cook1994}
\bibinfo{author}{Cook, J.R.}, \bibinfo{author}{Stefanski, L.A.},
  \bibinfo{year}{1994}.
\newblock \bibinfo{title}{Simulation-extrapolation estimation in parametric
  measurement error models}.
\newblock \bibinfo{journal}{Journal of the American Statistical Association}
  \bibinfo{volume}{89}, \bibinfo{pages}{1314--1328}.
\bibitem[{Curley(2017)}]{brenna2017}
\bibinfo{author}{Curley, B.}, \bibinfo{year}{2017}.
\newblock \bibinfo{title}{Nonlinear models with measurement error: Application
  to vitamin D}.
\newblock Ph.D. thesis. \bibinfo{address}{Iowa State University}.
\bibitem[{Gould et~al.(1999)Gould, Stefanski and Pollock}]{gould1999}
\bibinfo{author}{Gould, W.}, \bibinfo{author}{Stefanski, L.},
  \bibinfo{author}{Pollock, K.}, \bibinfo{year}{1999}.
\newblock \bibinfo{title}{Use of simulation–extrapolation estimation in
  catch–effort analyses}.
\newblock \bibinfo{journal}{Canadian Journal of Fisheries and Aquatic Sciences}
  \bibinfo{volume}{56}, \bibinfo{pages}{1234--1240}.
\bibitem[{Hardin et~al.(2003)Hardin, Schmiediche and Carroll}]{hardin2003}
\bibinfo{author}{Hardin, J.W.}, \bibinfo{author}{Schmiediche, H.},
  \bibinfo{author}{Carroll, R.J.}, \bibinfo{year}{2003}.
\newblock \bibinfo{title}{The simulation extrapolation method for fitting
  generalized linear models with additive measurement error}.
\newblock \bibinfo{journal}{The Stata Journal} \bibinfo{volume}{3},
  \bibinfo{pages}{373--385}.
\bibitem[{Hwang and Huang(2003)}]{hwang2003}
\bibinfo{author}{Hwang, W.}, \bibinfo{author}{Huang, S.Y.},
  \bibinfo{year}{2003}.
\newblock \bibinfo{title}{Estimation in capture‐recapture models when
  covariates are subject to measurement errors}.
\newblock \bibinfo{journal}{Biometrics} \bibinfo{volume}{59},
  \bibinfo{pages}{1113--1122}.
\bibitem[{Lin and Carroll(1999)}]{lin1999}
\bibinfo{author}{Lin, X.}, \bibinfo{author}{Carroll, R.J.},
  \bibinfo{year}{1999}.
\newblock \bibinfo{title}{Simex variance component tests in generalized linear
  mixed measurement error models}.
\newblock \bibinfo{journal}{Biometrics} \bibinfo{volume}{55},
  \bibinfo{pages}{613--619}.
\bibitem[{Mallick et~al.(2002)Mallick, Fung and Krewski}]{mallick2002}
\bibinfo{author}{Mallick, R.}, \bibinfo{author}{Fung, K.},
  \bibinfo{author}{Krewski, D.}, \bibinfo{year}{2002}.
\newblock \bibinfo{title}{Adjusting for measurement error in the cox
  proportional hazards regression model}.
\newblock \bibinfo{journal}{Journal of cancer epidemiology and prevention}
  \bibinfo{volume}{7}, \bibinfo{pages}{155--164}.
\bibitem[{Ponzi et~al.(2019)Ponzi, Keller and Muff}]{ponzi2019}
\bibinfo{author}{Ponzi, E.}, \bibinfo{author}{Keller, L.F.},
  \bibinfo{author}{Muff, S.}, \bibinfo{year}{2019}.
\newblock \bibinfo{title}{The simulation extrapolation technique meets ecology
  and evolution: A general and intuitive method to account for measurement
  error}.
\newblock \bibinfo{journal}{Methods in Ecology and Evolution}
  \bibinfo{volume}{10}, \bibinfo{pages}{1734--1748}.
\bibitem[{Sevilimedu et~al.(2019)Sevilimedu, Yu, Samawi and
  Rochani}]{sevil2019}
\bibinfo{author}{Sevilimedu, V.}, \bibinfo{author}{Yu, L.},
  \bibinfo{author}{Samawi, H.}, \bibinfo{author}{Rochani, H.},
  \bibinfo{year}{2019}.
\newblock \bibinfo{title}{Application of the misclassification simulation
  extrapolation procedure to log-logistic accelerated failure time models in
  survival analysis}.
\newblock \bibinfo{journal}{Journal of Statistical Theory and Practice}
  \bibinfo{volume}{13}, \bibinfo{pages}{1--16}.
\bibitem[{Staudenmayer and Ruppert(2004)}]{stauden2004}
\bibinfo{author}{Staudenmayer, J.}, \bibinfo{author}{Ruppert, D.},
  \bibinfo{year}{2004}.
\newblock \bibinfo{title}{Local polynomial regression and
  simulation–extrapolation}.
\newblock \bibinfo{journal}{Journal of the Royal Statistical Society: Series B
  (Statistical Methodology)} \bibinfo{volume}{66}, \bibinfo{pages}{17--30}.
\bibitem[{Stefanski and Bay(1996)}]{stef1996}
\bibinfo{author}{Stefanski, L.}, \bibinfo{author}{Bay, J.},
  \bibinfo{year}{1996}.
\newblock \bibinfo{title}{Simulation extrapolation deconvolution of finite
  population cumulative distribution function estimators}.
\newblock \bibinfo{journal}{Biometrika} \bibinfo{volume}{83},
  \bibinfo{pages}{407--417}.
\bibitem[{Stefanski and Cook(1995)}]{stefan1995}
\bibinfo{author}{Stefanski, L.A.}, \bibinfo{author}{Cook, J.R.},
  \bibinfo{year}{1995}.
\newblock \bibinfo{title}{Simulation-extrapolation: The measurement error
  jackknife}.
\newblock \bibinfo{journal}{Journal of the American Statistical Association}
  \bibinfo{volume}{90}, \bibinfo{pages}{1247--1256}.
\bibitem[{Stoklosa et~al.(2016)Stoklosa, Dann, Huggins and
  Hwang}]{stoklosa2016}
\bibinfo{author}{Stoklosa, J.}, \bibinfo{author}{Dann, P.},
  \bibinfo{author}{Huggins, R.M.}, \bibinfo{author}{Hwang, W.H.},
  \bibinfo{year}{2016}.
\newblock \bibinfo{title}{Estimation of survival and capture probabilities in
  open population capture–recapture models when covariates are subject to
  measurement error}.
\newblock \bibinfo{journal}{Computational Statistics \& Data Analysis}
  \bibinfo{volume}{96}, \bibinfo{pages}{74--86}.
\bibitem[{Wang et~al.(2009)Wang, Sun and Fan}]{wang2009}
\bibinfo{author}{Wang, X.}, \bibinfo{author}{Sun, J.}, \bibinfo{author}{Fan,
  Z.}, \bibinfo{year}{2009}.
\newblock \bibinfo{title}{Deconvolution density estimation with heteroscedastic
  errors using simex.} \URLprefix \url{https://arxiv.org/abs/0902.2117}.
\bibitem[{Wang et~al.(2010)Wang, Fan and Wang}]{wang2010}
\bibinfo{author}{Wang, X.F.}, \bibinfo{author}{Fan, Z.}, \bibinfo{author}{Wang,
  B.}, \bibinfo{year}{2010}.
\newblock \bibinfo{title}{Estimating smooth distribution function in the
  presence of heteroscedastic measurement errors}.
\newblock \bibinfo{journal}{Computational statistics \& data analysis}
  \bibinfo{volume}{54}, \bibinfo{pages}{25--36}.

\end{thebibliography}

\end{document}